\documentclass{LMCS}

\def\dOi{10(4:1)2014}
\lmcsheading%
{\dOi}
{1--33}
{}
{}
{Oct.~23, 2012}
{Oct.~28, 2014}
{}

\ACMCCS{[{\bf Theory of computation}]: Logic---Finite Model Theory}

\setlist[enumerate]{font=\normalfont,labelindent=*,leftmargin=*,start=1,label=(\roman*)}

\usepackage{hyperref}

\usepackage{calc}
\usepackage{ifthen}
\usepackage{epsfig}
\usepackage{graphicx}
\usepackage{latexsym}
\usepackage{amssymb,amsmath}

\theoremstyle{plain}\newtheorem{theorem}{Theorem}

\newtheorem{definition}[theorem]{Definition}

\newenvironment{renumerate}{\begin{enumerate}}{\end{enumerate}}

\renewcommand{\tilde}{\widetilde}
\renewcommand{\bar}{\overline}
\newcommand{\dd}{\mathrm{D}}

\newcommand{\SOL}{\mathrm{SOL}}
\newcommand{\FOL}{\mathrm{FOL}}
\newcommand{\MSOL}{\mathrm{MSOL}}
\newcommand{\CMSOL}{\mathrm{CMSOL}}
\newcommand{\CFOL}{\mathrm{CFOL}}
\newcommand{\IFPL}{\mathrm{IFPL}}
\newcommand{\FPL}{\mathrm{FPL}}

\newcommand{\N}{{\mathbb N}}

\newcommand{\fA}{{\mathfrak A}}
\newcommand{\fB}{{\mathfrak B}}

\newcommand{\Z}{{\mathbb Z}}

\newcommand{\Q}{{\mathbb Q}}
\newcommand{\bP}{{\mathbf P}}
\newcommand{\bNP}{{\mathbf{NP}}}

\newcommand{\cL}{\mathcal{L}}
\newcommand{\cP}{\mathcal{P}}

\renewcommand{\L}{\cL}

\begin{document}

\title[Connection Matrices and the Definability of Graph Parameters]{Connection Matrices and the Definability of Graph Parameters}

\author[T.~Kotek]{Tomer Kotek\rsuper a}
\address{{\lsuper a}Faculty of Informatics \\ 
Vienna University of Technology \\ 
Vienna, Austria}
\email{kotek@forsyte.at}

\thanks{{\lsuper a}The first author was partially supported by the Fein Foundation, the Graduate School of the Technion--Israel
Institute of Technology, the Austrian National Research
    Network S11403-N23 (RiSE) of the Austrian Science Fund (FWF) and
    the Vienna Science and Technology Fund (WWTF) through
    grants PROSEED, ICT12-059, and VRG11-005.
}

\author[J.~A.~Makowsky]{Johann A. Makowsky\rsuper b}
\address{{\lsuper b}Faculty of Computer Science \\
Technion--Israel Institute of Technology \\
Haifa, Israel}	
\email{janos@cs.technion.ac.il}  
\thanks{{\lsuper b}The second author was
partially supported by the Israel Science Foundation for the project ``Model Theoretic Interpretations
of Counting Functions'' (2007-2011) and the Grant for Promotion of Research by the
Technion--Israel Institute of Technology.}	




\begin{abstract}
In this paper we extend and prove in detail the Finite Rank Theorem
for connection matrices of graph parameters definable in 
Monadic Second Order Logic with counting ($\CMSOL$) from 
B. Godlin, T. Kotek and J.A. Makowsky (2008) and J.A. Makowsky (2009).
We demonstrate its vast applicability in simplifying known and 
new non-definability results of graph properties and finding new non-definability
results for graph parameters. 
We also prove a Feferman-Vaught Theorem for the logic CFOL, First Order Logic with
the modular counting quantifiers. 
\end{abstract}

\maketitle


\newcommand{\angl}[1]{\left\langle #1 \right\rangle}
\newcommand{\card}[3]{card_{\mathcal{#1},\bar{#2}}(#3(\bar{#2}))}
\newif\ifshort
\shorttrue
\newif\ifskip
\skiptrue


\section{Introduction}

\subsection*{Difficulties in proving non-definability}

Proving that a graph property $P$ is not definable in First Order Logic $\FOL$
can be a challenging task, especially on graphs with an additional linear order on the vertices.
Proving that a graph property such as 3-colorability, which is definable in
Monadic Second Order Logic $\MSOL$, is not definable in fixed point logic on ordered graphs
amounts to solving the famous $\bP \neq \bNP$ problem.

In the case of $\FOL$ and $\MSOL$ properties the basic tools for proving non-definability
are the various {\em Ehrenfeucht-Fra\"iss\'e games} also called  {\em pebble games}.
However,  proving the existence of winning strategies for these games can be
exasperating. Two additional tools can be used to make the construction of such winning strategies
easier and more transparent: the {\em composition of winning strategies} and the use of
{\em locality properties} such as Hanf locality and Gaifman locality. 
These techniques are by now well understood, even if not always simple
to apply, and are described in monographs such as \cite{bk:FMT,bk:Libkin2004}.
However these techniques are not easily applicable for stronger logics, such as
$\CFOL$ and $\CMSOL$, which extend $\FOL$ respectively $\MSOL$ with
modular counting quantifiers $\dd_{m,i}x \phi(x)$ which say that the number of elements satisfying $\phi$
equals $i$ modulo ${m}$. 
Furthermore, the pebble game method or the locality method may be difficult to use when dealing with
ordered structures or when proving non-definability for the case where the definition may use an order relation on the
universe in an {\em order-invariant} way.
Definability in $\MSOL$ and $\CMSOL$ depends on whether the vocabulary of graphs or hypergraphs is considered. 
The vocabulary of graphs is $\tau_{G}$ consisting of the edge relation symbol $E$. The vocabulary $\tau_{HG}$ of hypergraphs is 
two-sorted, with two unary relation symbols $V$ and $E$ whose interpretations partition the universe, and has
the binary incidence relation $R_{inc}$ between $V$ elements and $E$ elements. Structures of $\tau_{HG}$ can 
be not only graphs but hypergraphs. In \cite{bk:CourcelleEngelfriet2012} $\MSOL$ over $\tau_{G}$ is denoted
$MS_1$, while $\MSOL$ over $\tau_{HG}$ is denoted $MS_2$. In the case of $\FOL$ the choice of $\tau_G$ or $\tau_{HG}$ 
does not have any effect on definablility.

The notion of definability was extended in  \cite{ar:ArnborgLagergrenSeese88,ar:ArnborgLagergrenSeese91}
to integer valued graph parameters, and in
\cite{ar:CourcelleMakowskyRoticsDAM,ar:MakowskyTARSKI,ar:MakowskyZilber2006,ar:KotekMakowskyZilber08,ar:KotekMakowskyZilber11}
to real or complex valued graph parameters
and graph polynomials. 
In \cite{ar:MakowskyTARSKI} and \cite{ar:KotekMakowskyZilber08} graph polynomials definable in $\MSOL$ respectively $\SOL$ were introduced.
The techniques of pebble games and locality do not lend themselves easily, or are not useful at all, for
proving non-definability in these cases.

We assume the reader is familiar with the basics of
finite model theory \cite{bk:FMT,bk:Libkin2004} and graph theory \cite{bk:bollobas99,bk:Diestel05}.

\subsection*{Connection matrices}

{\em Connection matrices} were introduced in  \cite{ar:FreedmanLovaszSchrijver07,ar:Lovasz07}
by M. Freedman, L. Lov\'asz and A. Schrijver
where they were used to characterize 
to characterize various partition functions based on graph homomorphisms, cf. also \cite{bk:Lovasz12}.
Let $f$ be a graph parameter whose image is in some field $\mathbb{F}$ such as the real or the complex numbers.
A $k$-connection matrix $M(\sqcup_k, f)$ is an infinite matrix, where the rows and columns are indexed
by  finite $k$-labeled graphs $G_i$ and the entry $M(\sqcup_k, f)_{i,j}$ is given by the
value of $f(G_i \sqcup_k G_j)$. Here $\sqcup_k$ denotes the {\em $k$-sum} operation on $G_i$ and $G_j$, i.e. the operation of taking the disjoint union of $G_i$ and $G_j$
and identifying vertices with the corresponding  $k$-many labels.

In \cite{ar:GodlinKotekMakowsky08} connection matrices were used to show that certain
graph parameters and polynomials are not $\MSOL$-definable.
The main result of \cite{ar:GodlinKotekMakowsky08} is the {\em Finite Rank Theorem},
which states that the connection matrices of $\CMSOL$-definable graph polynomials have finite rank.
Connection matrices and the Finite Rank Theorem were generalized in \cite{ar:Makowsky09} to matrices $M(\Box, f)$ where
$\Box$ is a binary operation on labeled graphs subject to a smoothness condition depending on the logic
one wants to deal with. However, very few applications of the Finite Rank Theorem were given.

\subsection*{Properties not definable in \texorpdfstring{$\CFOL$}{CFOL} and \texorpdfstring{$\CMSOL$}{CMSOL}}
The purpose of this paper lies in the demonstration that the Finite Rank Theorem is a {\em truly manageable tool}
for proving non-definability which leaves no room for hand-waving arguments. 
To make our point we discuss graph properties (not)-definable in $\CFOL$ and $\CMSOL$. 
We also discuss  the corresponding (non)-definability questions in $\CMSOL$
for graph parameters  and graph polynomials.
Although one can derive pebble games for these two logics, see e.g. \cite{ar:KolaitisV95,ar:Nurmonen00},
using them to prove non-definability may be very awkward.

Instead we use
a Feferman-Vaught-type Theorem for $\CFOL$ for disjoint unions and Cartesian products, Theorem \ref{th:FV-CFOL},
which seems to be new for the case of products. The corresponding theorem for disjoint unions, Theorem \ref{th:FV-CMSOL}(i),
for $\CMSOL$ was proven by B. Courcelle \cite{bk:courcelle,bk:CourcelleEngelfriet2012,ar:MakowskyTARSKI}.

The proof of the Finite Rank Theorem for these logics follows from the Feferman-Vaught-type theorems.
The details will be spelled out in  Section \ref{se:fol}.

With the help of the Finite Rank Theorem we give new and uniform proofs for the following: 
\begin{renumerate}
\item
Using connection matrices for various generalizations of the Cartesian product $\times_{\Phi}$ we prove
non-definability of the following properties in $\CFOL$ with the vocabulary of graphs $\angl{V,E, <}$ 
with linear order:
\begin{itemize}
\item
Forests,
bipartite graphs,
chordal graphs,
perfect graphs,
interval graphs,
block graphs (every biconnected component, i.e., every block, is a clique),
parity graphs (any two induced paths joining the same pair of vertices have the same parity);
\item
Trees, connected graphs;
\item
Planar graphs, cactus graphs (graphs in which any two
cycles have at most one vertex in common) and
pseudo-forests (graphs in which every connected component has at most one cycle);
\item
Bridgeless graphs, $k$-connected. 
\end{itemize}
The case of connected graphs was also shown undefinable in $\CFOL$ by. J. Nurmonen in \cite{ar:Nurmonen00}
using his version of the pebble games for $\CFOL$.
\item
Using connection matrices for various generalizations of the disjoint union $\sqcup_{\Phi}$ we prove
non-definability of the following properties in $\CMSOL$ with the vocabulary of graphs $\angl{V,E, <}$
with linear order:
\begin{itemize}
\item
Hamiltonicity (via cycles or paths), graphs having a perfect matching,  cage graphs (regular graphs with as few vertices as possible for
their girth), well-covered graphs (where every minimal vertex cover has the same size
as any other minimal vertex cover). 
Here $\sqcup_{\Phi}$ is the join operation $\bowtie$. 
\item
The class of graphs which have a spanning tree of degree at most $3$.
Here $\sqcup_{\Phi}$ is a modified join operation \cite[Remark 5.21, Page 350]{bk:CourcelleEngelfriet2012}.
\end{itemize}
\item
Using connection matrices for various generalizations of the disjoint union $\sqcup_{\Phi}$ we prove
non-definability  of the following properties in $\CMSOL$ with the vocabulary of hypergraphs $\angl{V,E; R, <}$
with linear order:
\begin{itemize}
\item 
Regular graphs and
bi-degree graphs.
\item Graphs with average degree at most $\frac{|V|}{2}$;
\item Aperiodic digraphs (where the greatest common divisor of the lengths of all
cycles in the graph is $1$);
\item Asymmetric (also called rigid) graphs (i.e. graphs which have no non-trivial automorphisms).
\end{itemize}
\end{renumerate}

\subsection*{Graph parameters and graph polynomials not definable in \texorpdfstring{$\CMSOL$}{CMSOL}}

A graph parameter is $\CMSOL$-definable if it is the evaluation of a
$\CMSOL$-definable graph polynomial. The precise definition of definability of graph polynomials
is given in Section \ref{se:solpol}. 
Most prominent graph polynomials turned out to be definable in $\CMSOL$, sometimes using a linear order
on the vertices in an order-invariant way,
among them the various Tutte polynomials, interlace polynomials, matching polynomials, and many more, cf.
\cite{ar:MakowskyZoo}.
This led the second author to express his belief 
in \cite{ar:MakowskyZoo} that 
all ``naturally occurring graph polynomials''
are $\CMSOL$-definable. However, in \cite{ar:GodlinKotekMakowsky08} it was shown,
using connection matrices, that the graph polynomial counting harmonious colorings
is not $\CMSOL$-definable. A vertex coloring is harmonious if each pair of colors appears at most
once at the two end points of an edge, cf. \cite{ar:Edwards97,ar:HopcroftKrishnamoorthy83}.
That this is indeed a graph polynomial was shown in \cite{ar:KotekMakowskyZilber11}.
However, the main thrust of \cite{ar:GodlinKotekMakowsky08} consists in showing that certain graph parameters
are not evaluations of the standard prominent graph polynomials.

In Section \ref{se:solpolx}, we use connection matrices to show that
many ``naturally occurring graph polynomials'' are not $\CMSOL$-definable.
All these examples count various colorings and are graph polynomials by \cite{ar:KotekMakowskyZilber11}.
The corresponding notion of coloring is studied extensively in the literature.

To illustrate this we show that the following graph polynomials are not $\CMSOL$-definable in the language of graphs:
\begin{itemize}
 \item The {\em chromatic polynomial} $\chi(G,k)$ which counts proper vertex $k$-colorings of $G$;
 \item For every fixed $t\in\N^+$, $\chi_{mcc(t)}(G,k)$ which 
        counts vertex $k$-colorings $f:V(G)\to[k]$ for which no color
	  induces a subgraph with a connected component of size larger than $t$;
 \item The {\em vertex-acyclic polynomial} $\chi_{v-acyclic}(G,k)$ which counts 
       proper vertex $k$-color\-ings $f : V (G) \to [k]$ such that there is no two colored cycle in G.      
 \end{itemize}

\noindent We show that the following graph polynomials are not even $\CMSOL$-definable in the language of hypergraphs:
\begin{itemize}
\item The {\em rainbow polynomial} $\chi_{rainbow}(G,k)$ which counts {\em path-rainbow connected $k$-colorings}, which are functions
   $c: E(G) \rightarrow [k]$ such that between any two vertices $u,v \in V(G)$ there exists a path
     where all the edges have different colors.
\item $\chi_{convex}(G,k)$ is the number of {\em convex colorings}, which are vertex $k$-colorings \linebreak $f:V(G)\to[k]$
	     such that every color induces a connected subgraph of G.
\item For every fixed $t\in\N^+$, $\chi_{t-improper}(G,k)$ which counts {\em $t$-improper colorings}. A $t$-improper coloring is function
$f : V (G) \to [k]$ for which every color induces a graph in which no vertex
has degree more than $t$. 
\item $\chi_{non-rep}(G,k)$ which counts {\em non-repetitive colorings}.
A function f : E(G) → [k] is a non-repetitive coloring if the sequence of colors on any path in $G$ is non-repetitive.
A sequence $a_1 ,\ldots, a_r$ is non-repetitive if there is no $i, j \geq 1$ such
that $(a_i,\ldots,a_{i+j-1})=(a_{i+j},\ldots,a_{i+2j-1})$. 
\item The {\em harmonious polynomial} $\chi_{harm}(G,k)$ which counts 
       proper vertex $k$-colorings $f : V (G) \to [k]$ such that for any two distinct edges $(u_1,u_2)$ and $(v_1,v_2)$, it holds that
       $\{f(u_1),f(u_2)\}\not=\{f(v_1),f(v_2)\}$. The underlying proof idea here is similar to that in \cite{ar:GodlinKotekMakowsky08}. We
       bring it here for completeness. 
\end{itemize}
	     
\noindent Path-rainbow connected colorings were introduced in \cite{ar:Chartrand2008} and their complexity was studied in  
\cite{ar:ChakrabortyFischerMatsliahYuster2008}.  
$mcc(t)$-colorings were studied in \cite{ar:ADV03}, \cite{ar:LMST08} and \cite{ar:Farrugia04}.
Note $\chi_{mcc(1)}(G,k)$ is the chromatic polynomial.
Convex colorings were studied for their complexity
e.g. in \cite{ar:MS07} and \cite{arXiv:GoodallNoble2008}.
From \cite{ar:KotekMakowskyZilber11} we get that $\chi_{rainbow}(G,k)$, $\chi_{mcc(t)}(G,k)$, and $\chi_{convex}(G,k)$ are 
graph polynomials with $k$ as the variable.
Acyclic vertex
colorings were introduced in \cite{ar:Grunbaum73} and A. V. Kostochka proved in 1978
in his thesis that it is NP-hard to decide for a given G and k if the
there exists an acyclic vertex coloring with at most k colors, see \cite{ar:AlonMcDiarmidReed1991}.
Acyclic edge colorings were studied e.g. in \cite{ar:AlonZaks02,ar:Skulrattanakulchai04,ar:BasavarajuC08}.
It is NP-hard to determine whether G is
t-improperly 2-colorable for any fixed positive t (even if G is planar),
cf. \cite{ar:CowenGoddardJesurum1997}.
Non-repetitive colorings were introduced in \cite{pr:AGHR02}. Their complexity was studied in \cite{ar:Manin2007}.
The minimal number of colors needed to color $G$ in a non-repetitive
way is called the {\em Thue number} of $G$.

Section \ref{se:solpolx} contains more examples of graph polynomials and graph parameters not definable in $\CMSOL$.

\subsection*{Outline of the paper}
\label{se:outline}

In Section \ref{se:reg}  we illustrate the use of connection matrices
in the case of regular languages. This serves as a ``warm-up'' exercise.
In Section \ref{se:fol} we introduce the general framework for connection matrices
of graph properties, i.e., boolean graph parameters, and of properties of general $\tau$-structures.
In Section \ref{se:limits} we spell out the advantages and limitations of the method
of connection matrices in proving non-definability.
In Section \ref{se:even} we give a proof of the Feferman-Vaught theorem for $\CFOL$. 
In Section \ref{se:nondef} we illustrate the use of connection matrices and the Finite Rank Theorem
for proving non-definability of properties.
In Section \ref{se:solpol} we recall the framework of definable graph polynomials and $\tau$-polynomials
and the corresponding definable numeric parameters. 
In Section \ref{se:cmsol-pol-fv} we prove 
a Feferman-Vaught-type theorem for $\CMSOL$-polynomials with respect to sum-like operations, which
is the main ingredient in Finite Rank Theorems for $\CMSOL$-polynomials, and in Section \ref{se:solpolx}
we show how to prove non-definability of many numeric graph invariants.
\section{Connection Matrices for Regular Languages}
\label{se:reg}

Our first motivating examples deal with regular languages and the operation of concatenation $\circ$.
By the well-known B\"uchi-Elgot-Trakhtenbrot Theorem, see \cite{bk:FMT,bk:Libkin2004},
a language $L \subseteq \Sigma^*$ is regular if and only if 
the class $S_L$ of ordered structures representing the words of $L$
is definable in $\MSOL$ (or equivalently in $\CMSOL$ or $\exists\MSOL$, the existential fragment of $\MSOL$).
The connection matrix $M(\circ, L)$ with columns and rows indexed by all words of $\Sigma^*$ is defined by
$ M(\circ, L)_{u,v}= 1$  iff $u \circ v \in L$. \footnote{Strictly speaking 
we should use the characteristic function of $L$ rather than $L$. We allow this slight abuse of notation to achieve simpler notation. }

The Myhill-Nerode Theorem, see \cite{bk:HU,bk:Harrison1978}, 
can be used to derive the following properties of $M(\circ, L)$:

\begin{prop}
\label{p:reg-1}
Let $L \subseteq \Sigma^*$ be a regular language.
\begin{renumerate}
\item
There is a finite partition $\{U_1, \ldots, U_k\}$ of $\Sigma^*$
such that the sub-matrices obtained from restricting $M(\circ, L)$ to\footnote{$M(\circ, L)^{[U_i, U_j]}$ denotes the submatrix of $M(\circ, L)$ with rows and columns corresponding to $U_i$ and $U_j$ respectively. }
$M(\circ, L)^{[U_i, U_j]}$ have constant entries.
\item
In particular, the infinite matrix $M(\circ, L)$ has finite rank over any field $\mathcal{F}$.
\item
$M(\circ, L)$ has an infinite sub-matrix of rank at most $1$.
\end{renumerate}
\end{prop}

\noindent Now we can also look  at counting functions and numeric parameters of words,
such as the length $\ell(w)$ of a word $w$ or the number of words $s_L(w)$ in a language $L$ which are (connected)
sub-words of a given word $w$. The corresponding connection matrices $M(\circ, \ell)$
and $M(\circ, s_L)$ defined by 
$ M(\circ, \ell)_{u,v}= \ell(u \circ v)$ and
$ M(\circ, s_L)_{u,v}= s_L(u \circ v)$ respectively do not satisfy (i) and (iii) above,
but still have finite rank.
On the other hand the function $m_L(w)$ which gives the maximal size of a  word in $L$  
which occurs as a connected sub-word in $w$ gives rise to connection matrix $M(\bar{\circ}, s_L)$
of infinite rank. Here $u \bar{\circ} v = u \circ a \circ v$ where $ a \not\in \Sigma$
and therefore $m_L(u \bar{\circ} v) = \max \{ m_L(u), m_L(v) \}$.

We can use these connection matrices to show that $L_1= \{0^n \circ 1^n : n \in \N \}$ is not regular, by noting
that the sub-matrix $M(\circ, L_1)$ with columns indexed by $0^n$ and rows indexed by $1^n$
has $0$ everywhere but in the diagonal, hence has infinite rank, contradicting (ii) of Proposition \ref{p:reg-1}.

The numeric parameters on words $\ell, s_L$ are $\MSOL$-definable: $\ell(w) = \sum_{u <_{in} w} 1$, 
where $u <_{in} w$ means that $u$ is a proper possibly empty initial segment of $w$.
Similarly, $s_L(w) = \sum_{u <_{sw} w} 1 $, where $u <_{sw} w$ denotes the relation $u$ is a connected sub-word of $w$.
We shall give a general definition of $\MSOL$-definable numeric parameter in Section \ref{se:solpol}.
But we state here already

\begin{prop}
\label{p:reg-2}
The connection matrices $M(\circ, f)$ and $M(\bar{\circ}, f)$ have finite rank, provided $f$ is $\MSOL$-definable.
\end{prop}
\begin{cor}
The function $m_L(w)$ is not $\MSOL$-definable.
\end{cor}

Connection matrices for concatenation of words are known
in Automata Theory as Hankel-matrices and were introduced in
\cite{ar:CarlylePaz1971}, see also \cite{bk:BerstelReutenauer} and \cite{bk:DrosteKuichVogler}.

\section{Connection Matrices for Properties: The Framework}
\label{se:fol}\label{se:framework}

Let $\tau$ be a purely relational
finite vocabulary which may include constant symbols
and may include a distinguished binary relation symbol for a linear order.
A {\em $\tau$-property} is a class of finite $\tau$-structures closed under $\tau$-isomorphisms.
If the context is clear we just speak of properties and isomorphisms.
We denote by $\SOL(\tau)$ the set of $\SOL$ formulas over $\tau$. 
A sentence is a formula without free variables.

Let $\mathcal{L}$ be a subset of $\SOL$.
$\cL$ is a {\em fragment of $\SOL$} if the following conditions hold:
\begin{renumerate}
\item
For every finite relational vocabulary $\tau$ the set of $\cL(\tau)$
formulas contains all the atomic $\tau$-formulas and is closed under the
boolean operations $\land, \lor, \neg$ and under renaming of relation and constant symbols.
\item
$\cL$ is
equipped with a notion of {\em quantifier rank} $qr:\L\to \mathbb{N}$ and we denote
by $\cL_q(\tau)$ the set of formulas of quantifier rank at most $q$.
The quantifier rank of atomic formulas is $0$. 
The quantifier rank of a Boolean combination of formulas $\alpha_1,\ldots,\alpha_t \in \cL(\tau)$ is
the maximum quantifier rank of $\alpha_1,\ldots,\alpha_t$. 
The quantifier rank is sub-additive under substitution of sub-formulas.
\item
The set of formulas of $\cL_q(\tau)$ with a fixed set of free variables
is, up to logical equivalence, finite.
\item
Furthermore, if $\phi(x)$ is a formula of $\cL_q(\tau)$
with $x$ a free variable of $\cL$, then
there is a formula $\psi$ logically equivalent to
$\exists x \phi(x)$ in $\cL_{q'}(\tau)$ with $q' \geq q+1$.
\end{renumerate}

\noindent Typical fragments are 
$\FOL$ and $\MSOL$.
$\CMSOL$ and
the fixed point logics
$\IFPL$ and $\FPL$ 
and their corresponding finite variable subsets
correspond to fragments of $\SOL$ if we replace the counting or fixed-point operators 
by their $\SOL$-definitions. 

For two $\tau$-structures 
$\mathfrak{A}$ and
$\mathfrak{B}$ 
we define the {\em equivalence relation of $\cL_q(\tau)$- non-distingui\-shability}, and
we write
$\mathfrak{A} \equiv_q^{\cL} \mathfrak{B}$,
if they satisfy the same  sentences from $\cL_q(\tau)$.

Let $s: \N \rightarrow \N$ be a function.
A binary operation $\Box$ between $\tau$-structures is called {\em ($s$,$\cL$)-smooth},
if for all $q \in \N$ whenever 
$\mathfrak{A}_1 \equiv_{q+s(q)}^{\cL} \mathfrak{B}_1$ and
$\mathfrak{A}_2 \equiv_{q+s(q)}^{\cL} \mathfrak{B}_2$ then 
$$
\mathfrak{A}_1 \Box \mathfrak{A}_2 \equiv_q^{\cL} 
\mathfrak{B}_1 \Box \mathfrak{B}_2.
$$
A binary operation $\Box$ between $\tau$-structures is {\em $\cL$-smooth}
if it is ($0$,$\cL$)-smooth.

For two $\tau$-structures $\fA$ and $\fB$, we denote by 
$\fA \sqcup \fB$ the {\em disjoint union}, which is a $\tau$-structure\footnote{
The standard definition of disjoint union e.g. in \cite{bk:FMT} is for relational structures, however
we also allow contant symbols.  
For two relational $\tau$-structures $\fA$ and $\fB$ and tuples $\bar{a}$ and $\bar{b}$ 
of $A$ respectively $B$ elements, we denote by 
$\left\langle \fA,\bar{a}\right\rangle \sqcup \left\langle \fB,\bar{b}\right\rangle$
the disjoint union $\fA \sqcup \fB$ extended with the tuples $\bar{a}$ and $\bar{b}$, i.e. the structure
 $\left\langle \fA \sqcup \fB, \bar{a},\bar{b} \right\rangle$.
}; 
$\fA \sqcup_{rich} \fB$ the {\em rich disjoint union} which is the
disjoint union augmented with two unary predicates for the universes
$A$ and $B$ respectively;
$\fA \times \fB$ the {\em Cartesian product}, which is a $\tau$-structure;
and for graphs $G,H$ by $G \bowtie H$ the {\em join} of two graphs obtained from
the disjoint union of $G$ and $H$ by adding all possible edges between vertices of $G$ and
vertices of $H$. 

A {\em $\cL$-transduction} of $\tau$-structures into $\sigma$-structures is given
by defining a $\sigma$-structure inside a given $\tau$-structure.
The universe of the new structure may be a definable subset of an $m$-fold Cartesian product of
the old structure. If $m=1$ we speak of {\em scalar} and otherwise of {\em vectorized} transductions.
For every $k$-ary relation symbol $R \in \sigma$ we need a $\tau$-formula in $k \cdot m$ free individual 
variables to define it. We denote by $\Phi$ a sequence of $\tau$-formulas which defines a transduction.
We denote by $\Phi^{\star}$ the map sending $\tau$-structures into $\sigma$-structures induced by $\Phi$.
We denote by $\Phi^{\sharp}$ the map sending $\sigma$-formulas into $\tau$-formulas induced by $\Phi$.
For a $\sigma$-formula 
$\Phi^{\sharp}(\theta)$ is the {\em backward translation} of $\theta$ into a $\tau$-formula.
$\Phi$ is {\em quantifier-free} if all its formulas are from $\FOL_0(\tau)$.
We skip the details, and refer the reader to \cite{bk:Libkin2004,ar:MakowskyTARSKI}.

A fragment $\cL$ is {\em  closed under scalar transductions}, if for $\Phi$ such that all the formulas of $\Phi$
are in $\cL(\tau)$, $\Phi$ scalar,  and $\theta \in \cL(\sigma)$, the backward substitution
$\Phi^{\sharp}(\theta)$ is also in $\cL(\tau)$.
A fragment of $\SOL$ is called {\em tame} if it is 
closed under scalar transductions and containment of the form $\forall x (\varphi(x)\to \psi(x))$. 
$\FOL$,  $\MSOL$ and $\CMSOL$ are all tame fragments. So are their finite variable versions.

$\FOL$ and $\SOL$  are also closed under vectorized transductions, but the monadic fragments
$\MSOL$ and $\CMSOL$ are not.

We shall frequently use the following:
\begin{prop}
\label{p:fund}
Let $\Phi$ define a $\cL$-transduction from $\tau$-structures to $\sigma$-structures 
where  each formula is of quantifier rank at most $q$.
Let $\theta$ be a  $\cL(\sigma)_r$-formula.
Then
$$ \Phi^{\star}(\mathfrak{A}) \models \theta \mbox{   iff   } \mathfrak{A} \models \Phi^{\sharp}(\theta)$$
and $ \Phi^{\sharp}(\theta)$ is in $\cL(\tau)_{q+r}$.
\end{prop}

\begin{prop}[Smooth operations]
\label{p:fol-1}
\label{th:FV-CMSOL}
\ 
\begin{renumerate}
\item
The rich disjoint union $\sqcup_{rich}$ of $\tau$-structures and therefore also the disjoint union
are $\FOL$-smooth,
$\MSOL$-smooth and $\CMSOL$-smooth. They are not $\SOL$-smooth.
\item 
The $k$-sum $\sqcup_k$ of $\tau$-structures is $\MSOL$-smooth and $\CMSOL$-smooth, but not $\SOL$-smooth, for $k\in\mathbb \N$. 
\item
The Cartesian product $\times$ of $\tau$-structures is $\FOL$-smooth, but not
$\MSOL$-smooth
\item The join $\bowtie$ of $\tau$-structures is $\FOL$-smooth and $\MSOL$- and $\CMSOL$-smooth in the vocabulary of graphs,
but is not $\MSOL$- and $\CMSOL$-smooth in the vocabulary of hypergraphs. 
\item
Let $\Phi$ be a quantifier-free  scalar transduction of $\tau$-structures into $\tau$-structures and let $\Box$ be
an $\cL$-smooth operation.
Then the operation $\Box_{\Phi}(\mathfrak{A}, \mathfrak{B}) = \Phi^{\star}(\mathfrak{A} \Box \mathfrak{B})$
is $\cL$-smooth.
If $\Phi$ has quantifier rank at most $k$, it is $(k, \cL)$-smooth.
\end{renumerate} 
\end{prop}
\begin{proof}[Sketch of proof]
(i) is shown for $\FOL$ and $\MSOL$ using the usual pebble games.
For $\CMSOL$ one can use Courcelle's version of the Feferman-Vaught Theorem for
$\CMSOL$, cf. \cite{bk:courcelle,bk:CourcelleEngelfriet2012,ar:MakowskyTARSKI}.
(iii) is again shown using  the pebble game for $\FOL$.
(v) follows from Proposition \ref{p:fund}.
The negative statements are well-known, but also follow from the developments in the sequel. 
(ii) and (iv) follow from (i) and (v). 
\end{proof}

\begin{thm}[Feferman-Vaught Theorem for $\CFOL$]
\label{th:FV-CFOL}
\ 
\begin{renumerate}
\item
The rich disjoint union $\sqcup_{rich}$ of $\tau$-structures, and therefore the disjoint union, too,
is $\CFOL$-smooth.
\item
The Cartesian product $\times$ of $\tau$-structures is $\CFOL$-smooth.
\end{renumerate}
\end{thm}
\begin{proof}[Sketch of proof]
The proof does not use pebble games, but Feferman-Vaught-type reduction sequences.
(i) can be proven using the same reduction sequences which are used in
\cite{bk:courcelle,bk:CourcelleEngelfriet2012}.
We will prove (ii) in Section \ref{se:even}. 
\end{proof}
To the best of our knowledge, (ii) of Theorem \ref{th:FV-CFOL} has not been stated in the literature before.

\begin{rem}
Theorem \ref{th:FV-CFOL}(ii)  is proven using modifications of the reduction sequences from
\cite[Theorem 1.6]{ar:MakowskyTARSKI}.
Reduction sequences are tuples of formulas $\left\langle \psi_1^{A},\ldots,\psi_m^{A},\psi_1^B,\ldots,\psi_n^{B} \right\rangle$ 
obtained from a formula $\phi$ so that the truth-value of $\fA \Box \fB$ on $\phi$ is a Boolean combination of
the truth-values of $\fA$ on $\psi_1^{A},\ldots,\psi_m^{A}$ and $\fB$ on $\psi_1^{B},\ldots,\psi_m^{B}$.
We call this a Feferman-Vaught Theorem, because our proof actually computes the
reduction sequences explicitly. 
One might also try to prove the theorem 
using the pebble games defined in \cite{ar:Nurmonen00}, but at least for the case
of the Cartesian product, the proof would be rather complicated and less transparent.
\end{rem}

\begin{thm}[Finite Rank Theorem for tame $\cL$, \cite{ar:GodlinKotekMakowsky08,ar:Makowsky09}]
\label{th:FRT-tame}
\ \\
Let $\cL$ be a tame fragment of $\SOL$.
Let $\Box$ be a binary operation between $\tau$-structures which is $\cL$-smooth.
Let $\mathcal{P}$ be a $\tau$-property which is definable
by a $\cL$-formula $\psi$ and  $M(\Box, \psi)$ be the connection matrix defined by
$$
M(\Box, \psi)_{\mathfrak{A},\mathfrak{B}}= 1 \mbox{   iff   } \mathfrak{A} \Box \mathfrak{B} \models \psi
\mbox{   and   } 0 \mbox{   otherwise  }.
$$
Then
\begin{renumerate}
\item
There is a finite partition $\{U_1, \ldots, U_k\}$ of the (finite) $\tau$-structures
such that the sub-matrices obtained from restricting $M(\Box, \psi)$ to
$M(\Box, \psi)^{[U_i, U_j]}$ have constant entries.
\item
In particular, the infinite matrix $M(\Box, \psi)$ has finite rank over any field $\mathcal{F}$.
\item
$M(\Box, \psi)$ has an infinite sub-matrix of rank at most $1$.
\end{renumerate}
\end{thm}
\begin{proof}[Sketch of proof]
(i) follows from the definition of a tame fragment and of 
smoothness and the fact that there are only finitely many formulas
(up to logical equivalence) in $\cL(\tau)_q$.
(ii) and (iii) follow from (i).
\end{proof}

\section{Merits and Limitations of Connection Matrices}
\label{se:limits}

\subsection*{Merits}
The advantages 
of the Finite Rank Theorem for tame $\cL$ 
in proving that a property is not definable in $\cL$  
are the following:
\begin{renumerate}
\item
Once the $\cL$-smoothness of a binary operation has been established,
proofs of non-definability become surprisingly simple and transparent.
One of the most striking examples is the fact that asymmetric (rigid) graphs are not
definable in $\CMSOL$, cf. Corollary \ref{co:ms2}.
\item
Many properties can be proven to be non-definable using the same or similar
sub-matrices, i.e., matrices with the same row and column indices. 
This is well illustrated in the examples of Section \ref{se:nondef}.
\end{renumerate}

\subsection*{Limitations}
The classical method of proving non-definability in $\FOL$ using pebble games
is complete in the sense that a property is $\FOL(\tau)_q$-definable iff the class of its models is closed under
game equivalence of length $q$.
Using pebble games one proves easily that the class of structures without any relations of even cardinality, $\mathrm{EVEN}$,
is not $\FOL$-definable. This cannot be proven using connection matrices in the following sense:

\begin{prop}
Let $\Phi$ be a quantifier-free transduction between $\tau$-structures and let 
$\Box_{\Phi}$ be the binary operation  on $\tau$-structures:
$$\Box_{\Phi}(\mathfrak{A}, \mathfrak{B}) = \Phi^{\star}(\mathfrak{A} \sqcup_{rich} \mathfrak{B}) $$
Then the connection matrix $M(\Box_{\Phi}, \mathrm{EVEN})$ satisfies the properties (i)-(iii)
of Theorem \ref{th:FRT-tame}.
\end{prop}

\section{Proof of the smoothness of Cartesian product in \texorpdfstring{$\CFOL$}{CFOL}} 
\label{se:even}
Recall that we denote by  $\dd_{m,i}$ the
modular counting quantifiers $\dd_{m,i}x \phi(x)$ which says that the number of elements satisfying $\phi$
equals $i$ modulo ${m}$. 

The proof of Theorem \ref{th:FV-CFOL}(ii) follows exactly the proof of Theorem 1.6 in \cite{ar:MakowskyTARSKI}, in which 
an analogous statement was proven for the ordered sum of structures. 
We spell out the changes needed in the proof from \cite{ar:MakowskyTARSKI} for the ordered product $\times$.

\begin{proof}[Proof of Theorem \ref{th:FV-CFOL}]
Given $G_1=\angl{V_1,E_2,<_1}$ and $G_2=\angl{V_2,E_2,<_2}$, their 
ordered product is given by $\angl{V_1,E_2} \times \angl{V_2,E_2}$ together with the lexicographic order $<$ on $V_1\times V_2$
induced by $<_1$ and $<_2$. 

The proof proceeds by induction and computes the reduction sequences of $\CFOL$-formulas for $\times$. 
Reduction sequences are tuples of formulas $\left\langle \psi_1^{1},\ldots,\psi_m^{1},\psi_1^2,\ldots,\psi_n^{2} \right\rangle$ 
obtained from a formula $\phi$ so that the truth-value of $G_1 \times G_2$ on $\phi$ is a Boolean combination of
the truth-values $b_1^1,\ldots,b_m^1$
of $G_1$ on $\psi_1^{1},\ldots,\psi_m^{1}$ and the truth-values $b_1^2,\ldots,b_m^2$
of $G_2$ on $\psi_1^{2},\ldots,\psi_m^{2}$.
The existence of reduction sequences directly implies smoothness. 

The reduction sequences and Boolean functions of the atomic relations are given as follows.
The difficult cases are those of the quantifiers. We discuss $\dd_{2,0}$ in detail. For all other 
cases we just state the reduction sequences and the corresponding Boolean formulas. 
See the proof of Theorem 1.6 in \cite{ar:MakowskyTARSKI} for a detailed discussion of $\exists x \phi$. 
\begin{description}
 \item[For $E(u,v)$]~\\
      Reduction sequence: $\angl{E_1(u,v),E_2(u,v)}$\\
      Boolean function: $b_1^1\land b_1^2$. 
 \item[For $u\approx v$]~\\
      Reduction sequence: $\angl{u\approx_1 v,u \approx_2 v}$\\
      Boolean function: $b_1^1\land b_1^2$. 
 \item[For $u<v$]~\\
      Reduction sequence: $\angl{u<_1 v, u\approx_1 v, u<_2 v, u\approx_2 v }$\\
      Boolean function: $b_1^1 \lor \left(b_2^1 \land b_1^2\right)$. 
\end{description}

Let $\Phi = \left\langle \phi_1^1,\ldots,\phi_m^1,\phi_1^2,\ldots,\phi_m^2\right\rangle$ 
and $\Psi = \left\langle \psi_1^1,\ldots,\psi_m^1,\psi_1^2,\ldots,\psi_m^2\right\rangle$
be reduction sequences for $\phi$ and $\psi$ respectively, and let $B_\phi(\bar{b})$ and 
$B_\psi(\bar{b'})$ be the corresponding Boolean functions with disjoint variables. . 

\begin{description}
\item[For $(\phi\land\psi)$]~\\
      Reduction sequences: $\angl{\Phi,\Psi}$\\
      Boolean function: $B_\phi(\bar{b}) \land B_\psi(\bar{b'})$. 
\item[For $\neg \phi$]~\\
      Reduction sequences: $\Phi$\\
      Boolean function: $\neg B_\phi(\bar{b})$. 
\end{description}

Now we turn to $\dd_{2,0}x \phi$. 
We look at $B_\phi(\bar{b})$ in disjunctive normal form: 
\[B_1 = \bigvee_{j\in J} C_j\]
with
\[
 C_j = C_j^A \land C_j^B
\]
and
\[
 C_j^A = 
\left(
\bigwedge_{i\in J(j,A,pos)} b^A_i
\bigwedge_{i\in J(j,A,neg)} \neg b^A_i
\right)
\]
and 
\[
 C_j^B = \left( 
\bigwedge_{i\in J(j,B,pos)} b^B_i
\bigwedge_{i\in J(j,B,neg)} \neg b^B_i
\right)\,.
\]

$B_1$ has $2m$ Boolean variables, $b_1^A,\ldots,b_m^A,b_1^B,\ldots,b_m^B$. 
We assume without loss of gene\-rality that every $C_j$ contains all of the variables. 
In other words, 
$\{1,\ldots,m\}\backslash J(j,A,pos)$ $=J(j,A,neg)$ and $\{1,\ldots,m\}\backslash J(j,B,pos)=J(j,B,neg)$
for every $j\in J$. 

Now let:
\[
\alpha^A_j = 
\left( 
\bigwedge_{i\in J(j,A,pos)} \phi^A_i(x)
\bigwedge_{i\in J(j,A,neg)} \neg \phi^A_i(x)
\right)
\]
and similarly
\[
\alpha^B_j = 
\left( 
\bigwedge_{i\in J(j,B,pos)} \phi^B_i(x)
\bigwedge_{i\in J(j,B,neg)} \neg \phi^B_i(x)
\right)
\]

Consider the formula $\dd_{2,0}\,x \phi$. $\fA \times \fB\models \dd_{2,0}\,x \phi$ iff the number of pairs $(a,b)\in A\times B$
such that $\left\langle\fA\times\fB, (a,b)\right\rangle \models \phi$ is even.
$\left\langle\fA\times\fB, (a,b)\right\rangle \models \phi$ iff $B_1$ holds for $\left\langle \fA,a\right\rangle$
and $\left\langle \fB,b\right\rangle$. 
Note that $\left\langle \fA,a\right\rangle$ satisfies at most one of the $\alpha^A_j$
and similarly for $\left\langle \fB,b\right\rangle$.

The number of pairs $(a,b)$ for which 
$B_1$ holds 
is even 
iff
the number of $C_j$ such that $C_j$ holds for an odd number of pairs $(a,b)$ is even.
This holds iff
the number of $C_j$ such that $C_j^A$ holds for an odd number of $a\in A$ and
$C_j^B$ holds for an odd number of $b\in B$ , is even.

Let 
\[
\beta^A_j = \dd_{2,0}\,x \alpha^A_j 
\mbox{ \ \ \ and \ \ \ }
\beta^B_j = \dd_{2,0}\,x \alpha^B_j 
\]
and let 
\[
P=\{ T\subseteq J \mid |T|\mbox { is even }\} \,.
\]
Finally we put 
\[
 B_{\dd_{2,0}}(\bar{c}) = \bigvee_{T\in P} \left(\bigwedge_{j\in T} (\neg c_j^A \land \neg c_j^B) \bigwedge_{j\not\in T} (c_j^A \lor  c_j^B)\right)
\]
where $c_j^A=1$ iff $\fA \models \beta_j^A$ and 
$c_j^B=1$ iff $\fB \models \beta_j^B$.
 
So, we have:
\begin{description}
 \item[For $\dd_{2,0} x \phi$]~\\
      Reduction sequence: $\left\langle \beta_1^A,\ldots,\beta_{m(J)}^A,\beta_1^B,\ldots,\beta_{m(J)}^B \right\rangle$\\
      Boolean function: $B_{\dd_{2,0}}(\bar{c})$.
 \item[For $\exists x \phi$]~\\
      Reduction sequence: $\angl{\theta^A_1,\ldots,\theta^{A}_{m(J)},\theta^B_1,\ldots,\theta^{B}_{m(J)}}$\\
      Boolean function: $\displaystyle{B_{\exists}(\bar{c}) = \bigvee_{j\in J} (c_j^A \land c_j^B)}$.

\end{description}
\end{proof}

A similar proof covers all quantifiers $\dd_{a,b}$. E.g. consider the quantifiers $\dd_{3,0},\dd_{3,1},\dd_{3,2}$. 
Here we set, for each $i\in\{0,1,2\}$:
\begin{eqnarray*}
 B_{\dd_{3,i}}(\bar{c}) = &\displaystyle{\bigvee_{T_{1,1},T_{1,2},T_{2,1},T_{2,2}}} \Bigg(
& \bigwedge_{j\in T_{1,1}} (c_{j,1}^A \land c_{j,1}^B) \land \bigwedge_{j\in T_{1,2}}(c_{j,1}^A \land c_{j,2}^B)  \land\\
& & \bigwedge_{j\in T_{2,1}} (c_{j,2}^A \land c_{j,1}^B) \land \bigwedge_{j\in T_{2,2}} (c_{j,2}^A \land c_{j,2}^B) \land  \\
& & \bigwedge_{j\in J\backslash(T_{1,1}\cup T_{1,2}\cup T_{2,1}\cup T_{2,2})} (c_{j,0}^A \lor c_{j,0}^B) \ \ \  \ \ \Bigg)
\end{eqnarray*}
where the outer $\bigvee$ is over tuples $(T_{1,1},T_{1,2},T_{2,1},T_{2,2})$ of disjoint subsets of $J$, which 
additionally satisfy that 
\[
 |T_{1,1}|+|T_{2,2}|+2|T_{1,2}|+2|T_{2,1}| \equiv i \mod{3}\,.
\]
For every $i\in\{0,1,2\}$, $c_{j,i}^A = 1$ iff $\fA \models \dd_{3,i}x \alpha_j^A$ 
and $c_{j,i}^B = 1$ iff $\fB \models \dd_{3,i}x \alpha_j^A$.  
The reduction sequence
 of $\dd_{3,i}x \phi$ is
\footnote{If we were interested in making the reduction sequence shorter, we could have omitted $\dd_{3,2}x \alpha_j^A$ and $\dd_{3,2}x \alpha_j^B$ 
and express them using $\dd_{3,0}x \alpha_j^A$, $\dd_{3,1}x \alpha_j^A$, $\dd_{3,0}x \alpha_j^B$ and $\dd_{3,1}x \alpha_j^B$,
in a similar way to our treatment of $\dd_{2,0}$ and $\dd_{2,1}$ in the proof above. 
}
\begin{eqnarray*}
 & & \left\langle \dd_{3,0}x \alpha_1^A, \dd_{3,1}x \alpha_1^A,\dd_{3,2} x\alpha_1^A,\ldots,\dd_{3,0}x 
\alpha_{m(J)}^A, \dd_{3,1}x \alpha_{m(J)}^A,\dd_{3,2} x\alpha_{m(J)}^A,\right. \\
 & &  \left. 
\ \,\dd_{3,0}x \alpha_1^B, \dd_{3,1}x \alpha_1^B,\dd_{3,2} x\alpha_1^B,\ldots, \dd_{3,0}x \alpha_{m(J)}^B, 
\dd_{3,1}x \alpha_{m(J)}^B, \dd_{3,2} x\alpha_{m(J)}^B
\right\rangle\,.
\end{eqnarray*}

\section{Proving Non-definability of Properties}
\label{se:nondef}

\subsection*{Non-definability on \texorpdfstring{$\CFOL$}{CFOL}}
We will prove non-definability in $\CFOL$ using Theorem \ref{th:FV-CFOL} for
Cartesian products combined with $\FOL$ transductions. It is useful to 
consider a slight generalization of the Cartesian product as follows. We 
add two constant symbols $start$ and $end$ to our graphs. In $G^1 \times G^2$
the symbol $start$ is interpreted as the pair of vertices $(v^1_{start},v^2_{start})$ from $G^1$ and $G^2$
respectively such that $v^i_{start}$ is the interpretation of $start_i$ (i.e. $start$ in $G^i$) for $i=1,2$. 

The transduction
$ \Phi_{sym}(x,y)=E_{D}(x,y)\lor E_{D}(y,x)$
transforms
a digraph $D=(V_{D},E_{D})$
into an
undirected graph whose edge relation is the symmetric
closure of the edge relation of the digraph.

The following transduction $\Phi_{F}'$ transforms
the Cartesian product of two directed graphs
$G^i=(V_{1},E_{1},v^i_{start},v^i_{end})$
with the two constants $start_{i}$ and $end_{i}$,
$i=1,2$ into a certain digraph. 
It is convenient to describe $\Phi_{F}'$ as a transduction of the two input graphs $G^1$ and $G^2$: 
\begin{eqnarray*}
\Phi_{F}\left(\left(v_{1},v_{2}\right),\left(u_{1},u_{2}\right)\right) 
& = &
\left(E_{1}(v_{1},u_{1})\land E_{2}(v_{2},u_{2})\right)\lor\\
&  &
 \left(\left(v_{1},v_{2}\right),\left(u_{1},u_{2}\right)\right)
=
\left(\left(start_{1},start_{2}\right),\left(end_{1},end_{2}\right)\right)
\end{eqnarray*}

Consider the transduction $\Phi_F$ obtained from $\Phi_F'$ by applying $\Phi_{sym}$
when the input graphs are directed paths $P^i_{n_i}$ of length $n_i$. 
The input graphs look like this:

\includegraphics{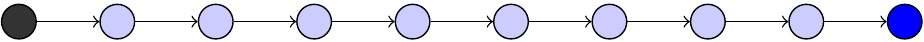}

The result of the application of the transduction is given in Figure \ref{fg:Forests}.
\begin{figure}[ht]
\begin{center}
\begin{tabular}{ccc}
\includegraphics[scale=0.8]{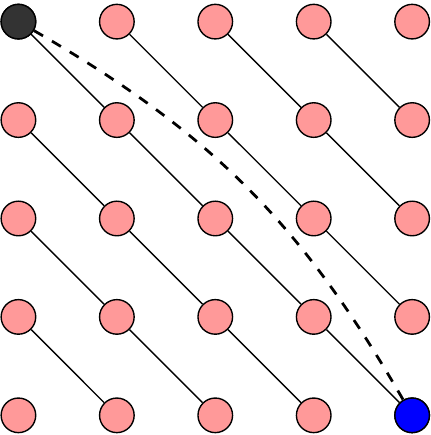} 
&  & \includegraphics[scale=0.8]{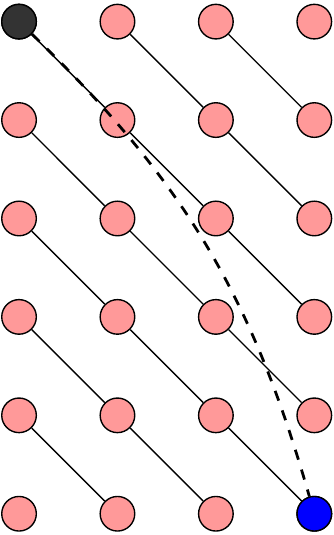}\tabularnewline
$n_{1}=n_{2}$ 
& ~~~~ &
$n_{1}\not=n_{2}$
\end{tabular}
 \caption{\label{fg:Forests} The result of applying $\Phi_F$ on the two directed paths of length $n_1$ and $n_2$.  
There is a cycle iff the two directed paths are of the same length. The black vertex corresponds to the two initial vertices
of the two directed paths. 
}
\end{center}
\end{figure}
The result of the transduction has a cycle iff $n_{1}=n_{2}$. 
The length of this cycle is $n_1$.
Hence,
the connection sub-matrix with rows and columns labeled by
directed paths of odd (even) length has ones on the main diagonal and zeros everywhere
else, so it has infinite rank.
Thus we have shown:

\begin{thm}
The graphs without cycles of odd (even) length are not $\CFOL$-definable even in the presence of a
linear order. 
\end{thm}

\begin{cor}Not definable in $\CFOL$ with order are:
\begin{enumerate}
\item Forests, bipartite graphs, chordal graphs, perfect graphs
\item interval graphs (cycles are not interval graphs)
\item Block graphs (every biconnected component 
is a clique)
\item Parity graphs (any two induced paths joining the same pair of vertices
have the same parity)
\end{enumerate}
\end{cor}

\noindent The transduction $\Phi_T$, obtained from $\Phi_T'$ below by additionally applying $\Phi_{sym}$, 
transforms the two directed paths into the structures in Figure \ref{fg:Trees}.
\begin{eqnarray*}
\Phi_{T}'\left(\left(v_{1},v_{2}\right),\left(u_{1},u_{2}\right)\right) 
& = & \left(E_{1}(v_{1},u_{1})\land E_{2}(v_{2},u_{2})\right)\lor\\
&  & \left(v_{1}=u_{1}=start_{1}\land E(v_{2},u_{2})\right)\lor\\
&  & \left(v_{1}=u_{1}=end_{1}\land E(v_{2},u_{2})\right),
\end{eqnarray*}

\begin{figure}[ht]
\begin{center}
\begin{tabular}{ccc}
\includegraphics[scale=0.8]{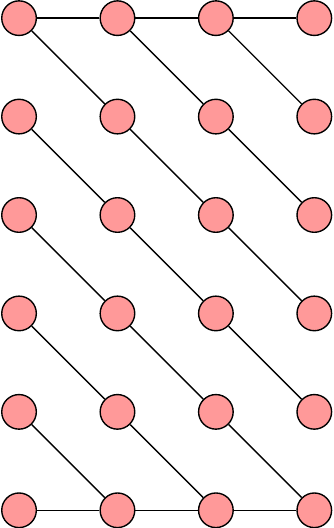}~~ 
& ~~\includegraphics[scale=0.8]{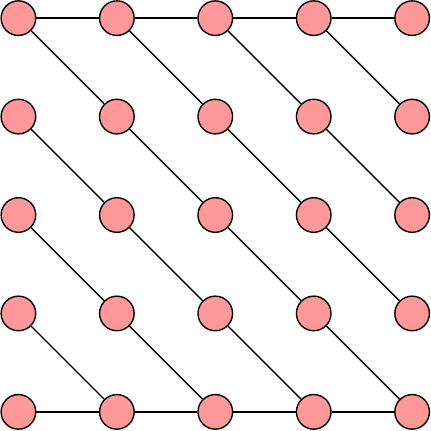}~~ &
 ~~\includegraphics[scale=0.8]{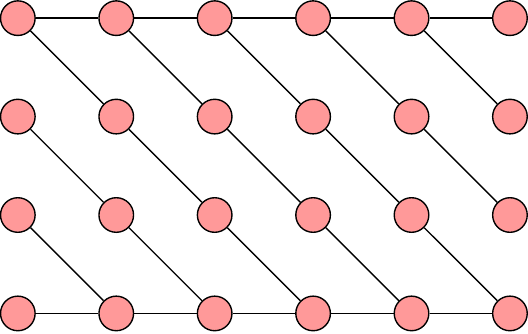}\tabularnewline
&  & \tabularnewline
$n_{1}>n_{2}$ & 
$n_{1}=n_{2}$ & 
$n_{1}<n_{2}$\tabularnewline
\end{tabular} \\
\caption{\label{fg:Trees} The result of applying $\Phi_T$ on two directed paths. 
We get a tree iff the two directed paths are of equal length. }
\end{center}
\end{figure}

So, the result of the transduction is a tree iff $n_{1}=n_{2}$. It
is connected iff $n_{1}\leq n_{2}$. Hence, both the connection matrices
with directed paths as row and column labels
of the property of being a tree and of connectivity have infinite
rank.

\begin{thm}
The properties of being a tree or a connected graph are not
$\CFOL$-definable even in the presence of linear order.
\end{thm}

The case of connected graphs was also shown non-definable in $\CFOL$ by. J. Nurmonen in \cite{ar:Nurmonen00}
using his version of the pebble games for $\CFOL$.


For our next connection matrix
we use the $2$-sum of the following two $2$-graphs. 
Recall that the $2$-sum $G\sqcup_2 H$ of two $2$-graphs 
$G$ and $H$ is obtained by taking the disjoint union of $G$ and $H$ and identifying
the corresponding labels. In this case the two labels are the constant symbols $start$ and $end$. 
\begin{enumerate}
\item the $2$-graph $(G,a,b)$ obtained from $K_{5}$ by choosing
two vertices $a$ and $b$ and removing the edge between them 
\item the symmetric closure of the Cartesian product of the two digraphs 
$P_{n_{1}}^{1}$ and $P_{n_{2}}^{2}$: 
\end{enumerate}
We denote this transduction by $\Phi_P$, see Figure \ref{fg:Planar}. 

\begin{figure}[ht]
\begin{center}
\begin{tabular}{ccc}
\includegraphics[scale=0.8]{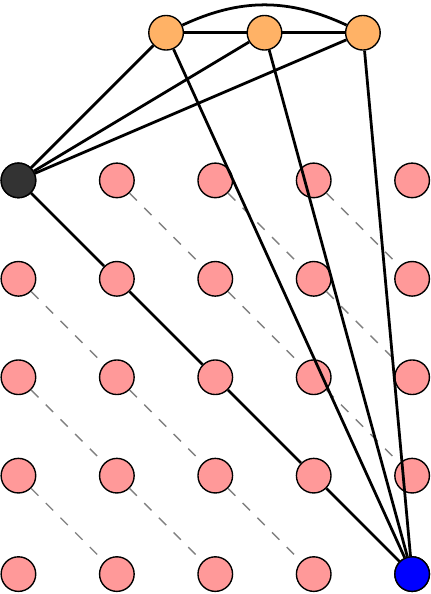}~~ 
&  & 
~~\includegraphics[scale=0.8]{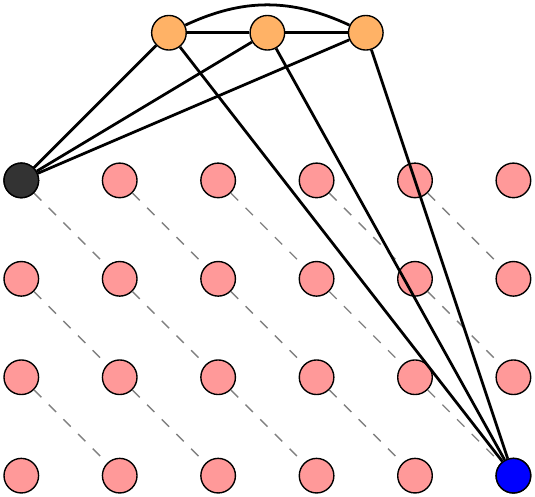}~~\tabularnewline
&  & \tabularnewline
$n_{1}=n_{2}$ &  & $n_{1}\not=n_{2}$\tabularnewline
\end{tabular} \\
\caption{\label{fg:Planar} The result of $\Phi_P$ on two directed paths. The graph obtained here is planar
iff the two directed paths are not of equal length. 
}
\end{center}
\end{figure}

So, the result of this construction has a clique of size $5$ as a
minor iff $n_{1}=n_{2}$. It can never have a $K_{3,3}$ as a minor.

\begin{thm}
\label{th:planar}
The class of planar graphs is
not $\CFOL$-definable on ordered graphs.
\end{thm}

If we modify the above construction by taking $K_3$ instead of $K_5$
and making \linebreak
$\left(start_{1},start_{2}\right)$ and $\left(end_{1},end_{2}\right)$
adjacent, we get 

\begin{cor}
The following classes of graphs
are not $\CFOL$-definable on ordered graphs.
\begin{renumerate}
\item
Cactus graphs, i.e. graphs in
which any two cycles have at most one vertex in common. 
\item Pseudo-forests, i.e. graphs in which every connected component has at
most one cycle.
\end{renumerate}
\end{cor}


\noindent For the next transduction we need to refer to the second and second to last vertices of $G^1$, both
of which are definable using $E_1$ and $start_1$ or $end_1$. The transduction $\Phi_{B}$ adds 
to the Cartesian product of $G^1$ and $G^2$ an edge between the second vertex in the first column and the 
second to last vertex in the last column. Moreover, $\Phi_{B}$ adds all the edges in the first and last rows and columns, 
except for two edges, the first edge in the first column and the last edge in the last column. 

The transduction $\Phi_B$, obtained from $\Phi_B'$ below by additionally applying $\Phi_{sym}$, 
transforms the two directed paths into the structures in Figure \ref{fg:Bridgeless}.
\begin{eqnarray*}
\Phi_{B}'\left(\left(v_{1},v_{2}\right),\left(u_{1},u_{2}\right)\right) 
& = & \left(E_{1}(start_1,v_1)\land E_{1}(u_1,end_1) \land (start_2=v_2) \land (end_2=u_2) \right)\lor \\
&   & \left(E_{1}(v_{1},u_{1})\land E_{2}(v_{2},u_{2})\right)\lor\\
&  & \left(v_{1}=u_{1}=start_{1}\land E(v_{2},u_{2})\right)\lor\\
&  & \left(v_{2}=u_{2}=start_{2}\land E(v_{1},u_{1})\land (v_1\not=start_1)\right)\lor\\
&  & \left(v_{1}=u_{1}=end_{1}\land E(v_{2},u_{2})\right)\lor\\
&  & \left(v_{2}=u_{2}=end_{2}\land E(v_{1},u_{1})\land (u_1\not=end_1)\right),
\end{eqnarray*}

\begin{figure}[ht]
\begin{center}
\begin{tabular}{ccc}
\includegraphics[scale=0.8]{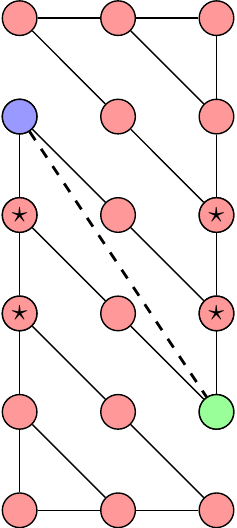}~~ 
& ~~\includegraphics[scale=0.8]{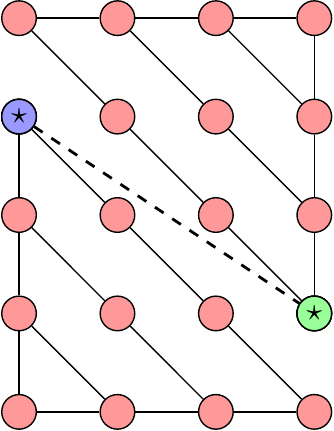}~~ &
 ~~\includegraphics[scale=0.8]{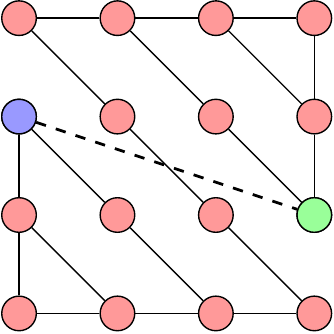}\tabularnewline
&  & \tabularnewline
$n_{1}>n_{2}+1$ & 
$n_{1}=n_{2}+1$ & 
$n_{1}< n_{2}+1$\tabularnewline
\end{tabular} \\
\caption{\label{fg:Bridgeless} The result of applying $\Phi_B$ on two directed paths. 
We always get a connected graph as a result of this transduction.
The resulting graph has articulation vertices (vertices whose removal leaves the graph disconnected) iff $n_1\geq n_2+1$. 
The same is true for the existence of bridges (edges whose removal leaves the graph disconnected).
We marked the articulation vertices with $\star$. The bridges are the edges between adjacent vertices marked with $\star$.
Note that if $n_1=n_2+1$ then there are two articulation vertices and one bridge, whereas if $n_1>n_2+1$ there are four articulation vertices 
and two bridges. 
}
\end{center}
\end{figure}

The result of the transduction has a bridge iff $n_{1}\geq n_{2}+1$. It
is $2$-connected iff $n_{1}\leq n_{2}$. Hence, the connection matrices
of  bridgelessness  and of $2$-connectivity
with directed paths as row and column labels
 have infinite
rank.

\begin{thm}
The properties of being bridgeless or $2$-connected are not
$\CFOL$-definable even in the presence of linear order.
\end{thm}

For any $\ell>0$, if we join $K_\ell$ to the result of $\Phi_B$
we get non-definability of $(\ell+2)$-connectivity, so we have:
\begin{cor}
For every $\ell\geq 0$, the property of being $(\ell+2)$-connected is not
$\CFOL$-definable even in the presence of linear order.
\end{cor}



\subsection*{Non-definability in \texorpdfstring{$\CMSOL$}{CMSOL}}
Considering the connection matrix where the rows and columns are labeled by
the graphs on $n$ vertices but without edges $E_n$,
the graph $E_{i}\Join E_{j}=K_{i,j}$ is
\begin{enumerate}
\item Hamiltonian iff $i=j$;
\item has a perfect matching iff $i=j$;
\item is a cage graph (a regular graph with as few vertices as possible for
its girth) iff $i=j$;
\item is a well-covered graph (every minimal vertex cover has the same size
as any other minimal vertex cover) iff $i=j$. 
\end{enumerate}
All of these connection matrices have infinite rank, so we get
\begin{cor}
\label{co:ms1}
None of the properties above are $\CMSOL$-definable as graphs even in the presence of an order.
\end{cor}
Using a modification $\widetilde{\Join}$ of the join operation from \cite[Remark 5.21, Page 350]{bk:CourcelleEngelfriet2012}
one can show the same
for the class of graphs which have a spanning tree of degree at most $3$.
Let $G$ and $H$ be graphs. Let $u_1,\ldots,u_m$ be the vertices of $H$. 
$G \widetilde{\Join} H$ is obtained from $G\Join H$ by taking two additional copies $H_1$ and $H_2$ of $H$ with vertices of the form $u_i^{b}$, $b=1$ or $b=2$ respectively,
and adding, for every $u_i\in V(H)$, the edges $(u_i,u_i^1)$ and $(u_i,u_i^2)$. 
If $G=\overline{K_n}$ and $H=\overline{K_m}$ are edgeless graphs with $n$ respectively $m$ vertices, then $G \widetilde{\Join} H$
has a spanning tree of degree at most $3$ iff $|V(G)|+2\geq |V(H)|$. The operation $\widetilde{\Join}$ is $\MSOL$-smooth over the vocabulary of graphs
since it can be defined as a transduction of the rich disjoint union of $G$ and $H$. The connection matrix of $G \widetilde{\Join} H$
and the property of having a spanning tree of degree at most $3$ has infinite rank, implying that the property is not $\MSOL$-definable
in the vocabulary of graphs. 

For any fixed natural number $d>3$, by performing a transduction on $G \widetilde{\Join} H$ which 
attaches $d-3$ new vertices as pendants to each vertex of $G \widetilde{\Join} H$, the non-definability result extends
to the class of graphs which have a spanning tree of degree at most $d$. 

For the language of hypergraphs we cannot use the join operation, since it is not smooth.
Note also that Hamiltonian and having a perfect matching are both definable in $\CMSOL$ in
the language of hypergraphs.
But using the connection sub-matrices of the disjoint union we still get:

\begin{enumerate}
\item Regular: 
$K_{i}\sqcup K_{j}$ is regular iff $i=j$;
\item 
A generalization of regular graphs are
{\em bi-degree} graphs, i.e., graphs where every vertex has one of two possible degrees.
$K_{i}\sqcup (K_{j}\sqcup K_{1})$  is a bi-degree graph iff $i=j$.
\item The average degree of $K_{i}\sqcup K_{j}$  is at most $\frac{|V|}{2}$ iff $i=j$;
\item A digraph is {\em aperiodic}  if the common denominator of the lengths of all
cycles in the graph is $1$. We denote by $C_i^d$ the directed cycle with $i$ vertices.
For prime numbers $p,q$ the digraphs
$C_{p}\sqcup C_{q}$ is aperiodic iff $p \neq q$.
\item A graph is {\em asymmetric (or rigid)} if it has no non-trivial automorphisms.
It was shown by P. Erd\"os and A. R\'enyi \cite{ar:ErdosRenyi63}
that almost all finite graphs are asymmetric. So there is an infinite set $I \subseteq \N$
such that for $i \in I$ there is  an asymmetric graph $R_i$ of cardinality $i$.
$R_{i}\sqcup R_{j}$ is asymmetric iff  $i \neq j$.
\end{enumerate}

\begin{cor}
\label{co:ms2}
None of the properties above are $\CMSOL$-definable as hypergraphs even in the presence of an order.
\end{cor}

\begin{rem} 
The case of asymmetric graphs illustrates that it is not always necessary to find explicit 
infinite families of graphs whose connection matrices are of infinite rank
in order to show that such a family exists. 
\end{rem}

\section{\texorpdfstring{$\cL$}{L}-Definable Graph Polynomials and Graph Parameters}
\label{se:sol-s}
\label{se:solpol}

\subsection*{\texorpdfstring{$\cL(\tau)$}{L(tau)}-polynomials}
Here we follow closely the exposition from \cite{ar:KotekMakowskyZilber11}.
Let $\cL$ be a tame fragment of $\SOL$.
We are now ready to introduce the $\cL$-definable polynomials. 
They are defined for $\tau$-structures and generally are called
{\em $\cL(\tau)$ invariants} as they map $\tau$-structures into some
commutative semi-ring
$\mathcal{R}$, which contains the semi-ring of the integers $\N$,
and are invariant under $\tau$-isomorphisms.
If $\tau$ is the vocabulary of graphs or hypergraphs,
we speak of graph invariants and graph polynomials.

For our discussion 
$\mathcal{R}=\N$ or
$\mathcal{R}=\Z$ suffices, but the definitions generalize.
Our polynomials have a fixed finite set of variables (indeterminates,
if we distinguish them from the variables of $\cL$), $\mathbf{X}$.

\begin{definition}[$\cL$-simple monomials]
\label{def:monomials-simple}
Let $\mathcal{M}$ be a $\tau$-structure.
Elements of $\N\cup \mathbf{X}$ are 
$\cL(\tau)$-definable
simple $\mathcal{M}$-monomials. 
\end{definition}

\begin{definition}[$\cL$-monomials]
\label{def:monomials}
Let $\mathcal{M}$ be a $\tau$-structure.
We define the  
$\cL(\tau)$-definable {\em $\mathcal{M}$-monomials}
inductively.
\begin{renumerate}
\item
$\cL(\tau)$-definable simple monomials are  
$\cL(\tau)$-definable
 monomials. 
\item
Let $\phi(a)$ be a $\tau \cup \{a\}$-formula in $\cL$,
where $a$ is a constant symbol not in $\tau$
and $U\in\tau$ is a unary relation symbol. 
Let $t$ be a simple monomial. Then
$$
\prod_{a: \langle \mathcal{M},a \rangle \models \phi(a)} t
$$
is a
$\cL(\tau)$-definable
$\mathcal{M}$-monomial.
\item
Finite products of $\cL(\tau)$-definable monomials and simple monomials are
$\cL(\tau)$-def\-inable
monomials. 
\end{renumerate}
\end{definition}

\noindent Note the degree of a monomial is polynomially
bounded by the cardinality of~$\mathcal{M}$.

\begin{definition}[$\cL$-polynomials]
\label{def:polynomials}
The
$\mathcal{M}$-polynomials definable in $\cL(\tau)$ are defined inductively: 
\begin{renumerate}
\item
$\cL(\tau)$-definable
monomials are
$\cL(\tau)$-definable
polynomials. 
\item
Let $\phi(\bar{a})$ be a $\tau \cup \{\bar{a}\}$-formula in $\cL$
where $\bar{a} = (a_1, \ldots , a_m)$ be a finite sequence of constant 
symbols not in $\tau$.
Let $t$ be a $\cL(\tau \cup \{\bar{a}\})$-polynomial. Then
$$
\sum_{\bar{a}: \langle \mathcal{M},\bar{a} \rangle \models \phi(\bar{a})} t
$$
is a $\cL(\tau)$-definable polynomial.
\item
Let $\phi(\bar{R})$ be a $\tau \cup \{\bar{R}\}$-formula in $\cL$
where $\bar{R} = (R_1, \ldots , R_m)$ be a finite sequence of relation 
symbols not in $\tau$.
Let $t$ be a $\cL(\tau \cup \{\bar{R}\})$-polynomial definable in $\cL$. Then
$$
\sum_{\bar{R}: \langle \mathcal{M},\bar{R} \rangle \models \phi(\bar{R})} t
$$
is a $\cL(\tau)$-definable polynomial.
\end{renumerate}
The polynomial $t$ may depend on relation or function
symbols  occurring in $\phi$.
\end{definition}
An $\mathcal{M}$-polynomial $p_{\mathcal{M}}(\mathbf{X})$
is an expression with parameter $\mathcal{M}$.
The family of polynomials, which we obtain from this expression
by letting $\mathcal{M}$ vary over all $\tau$-structures,
is called, by abuse of terminology, 
a $\cL(\tau)$-polynomial.

The quantifier rank of a $\cL(\tau)$-polynomial $p$ is the maximal quantifier rank of the formulas
defining $p$. 

Among the $\cL$-definable polynomials we find 
most of the known graph polynomials from the literature, 
cf. \cite{ar:MakowskyZoo,ar:KotekMakowskyZilber11}.
\begin{exa}[Matching polynomial] The matching generating polynomial of a graph $G$ is the generating function of matchings in $G$:
\[
 m(G)=\sum_{i=0}^{|E(G)|} m_i \mathbf{X}^i
\]
where $m_i$ is the number of matchings of size $i$ of $G$. $m(G)$ is an $\MSOL$-polynomial given by 
 \[
  \sum_{M\subseteq E(G): \left\langle G,M \right\rangle \models \varphi_{\mathit{match}}} \,\,
  \prod_{v:\left\langle G,M,v \right\rangle \models T} \mathbf{X}
 \]
 where $T$ is any tautology, $\varphi_{\mathit{sym}}(R) = \forall x \forall y (R(x,y)\to R(y,x))$, and 
 $\varphi_{\mathit{match}}(R) = \varphi_{\mathit{sym}}\land \forall x\forall y \forall z \neg (R(x,y)\land R(x,z))$.  
 The monomial $t=\prod_{v:\left\langle G,M,v \right\rangle \models T} \mathbf{X}$ depends on $M$. 
 We simplify the notation by writing 
\begin{gather}\label{eq:matching-pol}
 m(G,\mathbf{X})=\sum_{M \subseteq E(G): \varphi_{match}(M)} \mathbf{X}^{|M|}\,.
\end{gather}
\end{exa}

\begin{definition}[$\cL$-parameters and $\cL$-properties]
{\em $\cL$-definable numeric graph parameters} are evaluations of
$\cL$-definable polynomials and take values in $\mathcal{R}$.
{\em $\cL$-definable properties} are  special cases of numeric parameters
which have boolean values. 
\end{definition}

\begin{rem}[$\FOL$- and $\CFOL$-parameters]\label{rm:fol-param}
For $\FOL$- and $\CFOL$-parameters, the formulas must not contain any second order variables. Therefore, sums of the form 
$\sum_{\bar{R}: \langle \mathcal{M},\bar{R} \rangle \models \phi(\bar{R})} t$ are not allowed, 
while sums of the form $\sum_{\bar{a}: \langle \mathcal{M},\bar{a} \rangle \models \phi(\bar{a})} t$ are allowed. 
Moreover, the monomials are required to be simple monomials. 
\end{rem}


\subsection*{Sum-like and product-like operations}
For the proof of the
 Finite Rank Theorem for $\cL$-polynomials which involve second order variables
it is not enough that the binary operation $\Box$ on $\tau$-structures is $\cL$-smooth.
We need a way to uniquely decompose the relation over which we perform summation in $\fA \Box \fB$
into relations in $\fA$ and $\fB$ respectively, from which we can reconstruct the relation
in $\fA \Box \fB$.
For our discussion here it suffices to restrict $\Box$ to 
{\em $\cL$-sum-like operations}.
$\fA \Box \fB$ 
is {\em $\cL$-sum-like} if there is a {\em scalar} $\cL$-transduction $\Phi$
such that
$$\fA \Box \fB  = \Phi^{\star}(\fA \sqcup_{rich} \fB).$$
An operation is {\em $\cL$-product-like} if instead of scalar transductions we also
allow {\em vectorized} transductions.
Typically, the Cartesian product is $\FOL$-product-like, but not sum-like.
The $k$-sum and the join operation on graphs are $\FOL$-sum-like (but, in the case of join, not on hypergraphs). 
\begin{rem}[$\CFOL$-polynomials]
In Remark \ref{rm:fol-param} we required that $\CFOL$-parameters only have simple monomials. 
We did so because the Feferman-Vaught theorem and the Finite Rank theorem for $\times$ do not hold when allowing 
general monomials. 
Let $\CFOL$-polynomials be the extension of $\CFOL$-parameters by allowing monomials which are not simple. 

For example, consider the following $\CFOL$-polynomial, which is not a $\CFOL$-parameter:
$$
\prod_{a \in M} \mathbf{X} = \mathbf{X}^{|M|}
$$
The connection matrix $M(\times,\mathbf{X}^{|M|})$ restricted to $E_n$, the edgeless graphs, has entries $m_{i,j}=\mathbf{X}^{ij}$, is a Vandermonde matrix
and therefore has infinite rank.  

The graph polynomial $\mathbf{X}^{|M|}$ does not satisfy a bilinear reduction theorem 
such as Theorem \ref{th:fv-disjoint-union}. 
Assume there is a polynomial $Q$, and $\CFOL$-polynomials $p_1,\ldots,p_t$ such that
\[p(G_1 \times G_2)= Q(p_1(G_1),\ldots,p_t(G_1),p_1(G_2),\ldots,p_t(G_2))
\]
Let $G_1=G_2=E_n$ where $E_n$ is the edgeless graph with $n$ vertices.
By definition of $\CFOL$-polynomials, $p_i(E_n)$ has degree at most $n$ in $\mathbf{X}$. 
Therefore, $Q(p_1(G_1),\ldots,p_t(G_1),p_1(G_2),$ $\ldots,p_t(G_2))$ has degree at most $c n$ in  $\mathbf{X}$ for some $c\in \mathbb{N}$. 
In contrast, $p(E_n \times E_n)=\mathbf{X}^{n^2}$. Therefore, it cannot be the case that there is such $Q$.

\end{rem}

\subsection*{The Finite Rank Theorem for \texorpdfstring{$\cL$}{L}-polynomials}

Now we can state the Finite Rank Theorem for $\cL$-polynomials.
The proof uses the same techniques as in \cite{ar:CourcelleMakowskyRoticsDAM,ar:MakowskyTARSKI}.

\begin{thm}[The Finite Rank Theorem for $\cL$-polynomials]
\label{maintheorem}
\ \\
Let $\cL$ be a tame fragment of $\SOL$ such that $\sqcup_{rich}$ is $\L$-smooth
and $\Box$ be an $\cL$-sum-like operation between $\tau$-structures.
Let $P$ be a $\cL(\tau)$-polynomial.
Then the connection matrix $M(\Box, P)$ has finite rank.
\end{thm}

The proof is given in Section \ref{se:cmsol-pol-fv}
In \cite{ar:GodlinKotekMakowsky08} the theorem was only formulated for
$k$-sums, and the join operation and for the logic $\CMSOL$.

\begin{thm}[The Finite Rank Theorem for $\CFOL$-parameters]
\label{maintheorem-fo}
\ \\
Let $\boxtimes$ be an $\CFOL$-product-like operation between
$\tau$-structures.
Let $P$ be an $\CFOL(\tau)$-parameter.
Then the connection matrix $M(\boxtimes, P)$ has finite rank.
\end{thm}
The proof of Theorem \ref{maintheorem-fo} is similar to that of Theorem \ref{maintheorem}.


\newcommand{\HinSet}{\Theta}
\newcommand{\hinfor}{\theta}
\newcommand{\fM}{\mathfrak{M}}
\newcommand{\tp}{\tilde{p}}

\section{A Feferman-Vaught-type theorem for \texorpdfstring{$\cL$}{L}-polynomials and smooth operations}
\label{se:cmsol-pol-fv}
In this section we prove an analogue of the Feferman-Vaught theorem for sum-like operations, Theorem \ref{maintheorem}.
We start by explicitly giving a Feferman-Vaught theorem for the matching polynomial, 
proceed with the case of rich disjoint union, 
and end with all sum-like operations. 


\subsection{The matching polynomial}
As an example, we prove here a Feferman-Vaught theorem for the matching generating polynomial with the $1$-sum operation. 
$m(G,\mathbf{X})$ is a polynomial in $\Z[\mathbf{X}]$. As we need a field we work in the field of rational functions $\Q(\mathbf{X})$.

We will need the following auxiliary graph polynomials whose inputs are $1$-labeled graphs $(G,v_G)$:
\begin{eqnarray}
 m^+(G,v_G,\mathbf{X})&=&\sum_{M \subseteq E_G: \varphi_{match}(M)\land \varphi_{+}(M,v_G)} \mathbf{X}^{|M|}\\
 m^-(G,v_G,\mathbf{X})&=&\sum_{M \subseteq E_G: \varphi_{match}(M)\land \neg \varphi_{+}(M,v_G)} \mathbf{X}^{|M|}
\end{eqnarray}
where $\varphi_{+}$ says that $v_G$ is incident to an edge of $M$. 
The definition of $m(G,\mathbf{X})$ in Equation (\ref{eq:matching-pol}) extends to $1$-labeled graphs $(G,v_G)$ 
naturally by ignoring the label on $v_G$. 
For a $1$-labeled graph $(G,v_G)$, 
$m(G,\mathbf{X})=m(G,v_G,\mathbf{X})$ and 
\begin{gather}
  m(G,v_G,\mathbf{X})=m^{+}(G,v_G,\mathbf{X})+m^{-}(G,v_G,\mathbf{X})\,.
\end{gather}

\begin{thm}
Let $\mathfrak{m}(G, v_G;\mathbf{X}) =(m^{+}(G,v_G,\mathbf{X}),m^{-}(G,v_G,\mathbf{X}), m(G,v_G,\mathbf{X}))$.
\begin{renumerate}
\item
 There exists a polynomial $Q\in \Z[\mathbf{X}_1,\mathbf{X}_2,\mathbf{Y}_1,\mathbf{Y}_2]$ such that 
 for any two $1$-labeled graphs $(G,v_G)$ and $(H,v_H)$,
 \begin{gather}
 \begin{array}{lll}
  m((G,v_G)\sqcup_1 (H,v_H),\mathbf{X}) &=& Q\Big(m^+(G,v_G,\mathbf{X}),m^-(G,v_G,\mathbf{X}),\\
  && m^+(H,v_H,\mathbf{X}),m^-(H,v_H,\mathbf{X})\Big)\,.
  \end{array}
 \end{gather}
\item
There exist a matrix $A \in \Q(\mathbf{X})^{(3 \times 3)}$
such that $Q$ can be written as the bilinear form
 \begin{gather}
  m((G,v_G)\sqcup_1 (H,v_H),\mathbf{X}) = \mathfrak{m}(G, v_G,\mathbf{X}) \cdot A \cdot \mathfrak{m}(H, v_H,\mathbf{X})^{tr}.
 \end{gather}
\end{renumerate}
\end{thm}
\begin{proof}
(i):
A matching $M$ of $(G,v_G)\sqcup_1 (H,v_H)$ is a disjoint union of two matchings $M_G$ and $M_H$ of $G$ and $H$ respectively
 such that at most one of $M_G$ and $M_H$ is incident to $v_G$ respectively $v_H$.
 $m^{+}((G,v_G)\sqcup_1 (H,v_H))$ is the generating function of matchings $M_G\sqcup M_H$ such that 
 one of $M_G$ and $M_H$ is incident to $v_G$ respectively $V_H$.
 Analogously,
 $m^{-}((G,v_G)\sqcup_1 (H,v_H))$ is the generating function of matchings $M_G\sqcup M_H$ such that 
 both $M_G$ and $M_H$ are not incident to $v_G$ respectively $v_H$.  
 So we have:
\begin{eqnarray*}
m((G,v_G)\sqcup_1 (H,v_H),\mathbf{X}) &=& m^{+}(G,v_G,\mathbf{X}) \cdot m^{-}(H,v_H,\mathbf{X})+ \\ 
& & m^{-}(G,v_G,\mathbf{X}) \cdot m^{+}(H,v_H,\mathbf{X})+ \\
& & m^{-}(G,v_G,\mathbf{X}) \cdot m^{-}(H,v_H,\mathbf{X})\,.
\end{eqnarray*}
The theorem holds with the polynomial
\begin{gather}
 Q = \mathbf{X}_1 \mathbf{Y}_2 + \mathbf{X}_2 \mathbf{Y}_1 + \mathbf{X}_2 \mathbf{Y}_2 \,.
\end{gather}
(ii):
The matrix $A$ can be taken to be
\begin{gather}
A = \left( \begin{array}{ccc}
0 & 1 & 0  \\
1 & 1 & 0  \\
0 & 0 & 0  \\
\end{array} \right)
\end{gather}
\end{proof}
\begin{cor}
The connection matrix $M(\sqcup_1,m)$ has rank $2$. 
\end{cor}
\begin{proof}
The connection matrix $M(\sqcup_1,)$ is spanned by two rows which correspond to $m^+$ and $m^-$.
Hence, the rank of $M(\sqcup_1,m)$ is a t most $2$. 
\end{proof}

\begin{rem}
The matrix $A$ turns out to be a $(0,1)$-matrix due to a good choice of auxiliary graph polynomials.
If we replace $\mathfrak{m}(G, v_G,\mathbf{X})$ by
\begin{gather}
\mathfrak{m}_1(G, v_G,\mathbf{X})= \mathfrak{m}(G, v_G,\mathbf{X}) \cdot B, 
\end{gather}
where $B$ is any $(3 \times 3)$-matrix over $\Q(\mathbf{X})$,
we get
 \begin{gather}
  m((G,v_G)\sqcup_1 (H,v_H),\mathbf{X}) = \mathfrak{m}_1(G, v_G,\mathbf{X}) \cdot A_1  \cdot \mathfrak{m}_1(H, v_H,\mathbf{X})^{tr}
 \end{gather}
with $A_1 = B \cdot A \cdot B^{tr}$  where $A_1$ can have arbitrary entries from $\Q(\mathbf{X})$.
\end{rem}

\subsection{Rich disjoint union}\label{subse:proof-rich}

Here we prove a Feferman-Vaught-type theorem for rich disjoint union and structures of 
some vocabulary $\tau$. We do this for $\cL(\tau)$-polynomials of the form
\begin{equation}\label{eq:p}
 p(\fM)= \sum_{U\subseteq M: \Omega(U)} \left(\prod_{c \in M: \phi_\mathbf{X}(c,U)} \mathbf{X}\right) \left( \prod_{c \in M: \phi_\mathbf{Y}(c,U)} \mathbf{Y}\right)
\end{equation}
where $\Omega,\phi_\mathbf{X},\phi_\mathbf{Y}$ are $\cL$-sentences over the appropriate expansion of $\tau$. 
Generalization to all $\cL(\tau)$-polynomials is not hard 
using the normal form lemma for $\cL$-polynomials, Lemma \ref{lem:normal-form}:


\begin{lem}[Normal Form Lemma, \cite{ar:KotekMakowskyZilber11}]
 \label{lem:normal-form} Let $\cL$ be a fragment of $\SOL$. 
  Let $p$  be a $\cL(\tau)$-polynomial. 
  Then there exist $s,t\in \mathbb{N}$, $\cL$-formulas
  $\Omega$,   and
  $\phi_1,\ldots,\phi_t$,
  and $X_1,\ldots,X_t\in \mathbf{X}\cup \mathbb{N}$
  such that 
 \begin{equation}
  p(\fM) = 
  \sum_{U_1,\ldots,U_{s},a_1,\ldots,a_r:\Omega(\bar{U},\bar{a})}\,\,
  \left(\prod_{c \in M: \phi_1(c,\bar{U},\bar{a})} {X}_1\right)
  \cdots
  \left(\prod_{c \in M: \phi_t(c,\bar{U},\bar{a})} {X}_t\right)
  \end{equation}
  where $\bar{U}=U_1,\ldots,U_s$ and $\bar{a}=a_1,\ldots,a_r$. 
\end{lem}
The monomials of $\cL$-polynomials are already in the desired form. To prove Lemma \ref{lem:normal-form}, 
the main property of $\L$ that we use here is closure under conjunction which allows us to eliminate
nested sums. 



The following is well known \cite{bk:FMT}:
\begin{prop}[Hintikka sentences]
\label{prop:Theta}
Let $\L$ be a fragment of $\SOL$. 
Let $\tau$ be a vocabulary.
For every $q\in\N$ there is a finite set $\HinSet_{\tau,q}(\tau)$
of $\L(\tau)$-sentences of quantifier rank $q$ 
such that:
\begin{renumerate}
\item
every $\hinfor \in \HinSet_{\tau,q}(\tau)$ has a model;
\item
the conjunction of any two sentences $\hinfor_1, \hinfor_2 \in \HinSet_{\tau,q}(\tau)$
is not satisfiable;
\item
every $\L(\tau)$ sentence $\psi$ of quantifier rank at most $q$ is equivalent to exactly one 
finite disjunction of sentences in $\HinSet_{\tau,q}(\tau)$;
\item
every finite $\tau$-structure $\fA$ satisfies exactly one sentence $\theta_{\tau,q}(\fA)$ of $\HinSet_{\tau,q}(\tau)$. 
\end{renumerate}
\end{prop}

\noindent Proposition \ref{prop:Theta} uses that $\L$ is fragment and therefore $\L_q(\tau)$ is finite up to equivalence and is closed under Boolean operations.




The following proposition reformulating $\L$-smoothness follows directly from the definition of $\L$-smoothness in Section \ref{se:framework} and Proposition \ref{prop:Theta}:
\begin{prop} \label{prop:reform}
Let $\L$ be a fragment of $\SOL$ such that $\sqcup_{rich}$ is $\L$-smooth. 
Let $\tau_1$ and $\tau_2$ be a vocabularies and let $\sigma$ be the vocabulary of the rich disjoint union of a $\tau_1$-structure and $\tau_2$-structure.
There exists a function 
$g_{\tau_1,\tau_2}:\HinSet_{\tau_1,q} \times \HinSet_{\tau_2,q}\to \HinSet_{\sigma,q}$
such that for every two $\tau_1$-structures $\fA$ and $\fB$, 
\[
\theta_{\sigma,q}(\fA \sqcup_{rich} \fB) = g_{\tau_1,\tau_2}(\theta_{\tau,q}(\fA),\theta_{\tau_2}(\fB))\,. 
\]
We write $g_\tau$ rather than $g_{\tau_1,\tau_2}$ when $\tau = \tau_1 =\tau_2$. 
\end{prop}


In the proof of Theorem \ref{th:fv-disjoint-union} we will
use $\sigma$ with various subscripts to denote vocabularies 
which correspond to structures obtained as rich disjoint unions, and
$\tau$ with various subscripts to denote vocabularies of pre-union structures). 

\begin{thm}[Bilinear Reduction Theorem]\label{th:fv-disjoint-union}
Let $\L$ be a tame fragment of $\SOL$ such that $\sqcup_{rich}$ is $\L$-smooth. 
Let $\tau$ be a vocabulary and let $\sigma$ be the vocabulary of the rich disjoint union of $\tau$-structures.
Let $p$ be a $\L(\sigma)$-polynomial in the form of Equation (\ref{eq:p}). 
There exist $t\in \mathbb{N}$, $\L(\tau)$-polynomials $p_1,\ldots,p_{2t}$ 
and a polynomial $Q\in \Z[{w}_1,\ldots, {w}_t,{v}_1,\ldots,{v}_t]$ 
such that
\begin{renumerate}
 \item
  for any two $\tau$-structures $\fA$ and $\fB$
 \[
  p(\fA \sqcup_{rich} \fB) = Q\left(p_1(\fA),\ldots,p_t(\fA),p_{t+1}(\fB),\ldots,p_{2t}(\fB)
\right)\,.
 \]
 \item
  $Q(\bar{w}, \bar{v})$ can be written as a bilinear form
   $$\bar{w}^{tr} M \bar{v}$$
where the matrix $M$ has as entries polynomials with coefficients from
$\Z$ and is independent of the $\tau$-structures. 
\end{renumerate}
\end{thm}
\begin{rem}
In the following proof, the graph polynomials $p_i$ as well as $p$ all have the same quantifier rank $qr(p)$ and hence
$t$ can be taken as the number $\beta(qr(p))$ of $\L(\sigma)$-polynomials of quantifier rank at most $qr(p)$,
which is finite.
\end{rem}

\proof
Let $q$ be the maximum quantifier rank of $\Omega$, $\phi_\mathbf{X}$ and $\phi_\mathbf{Y}$. 

By definition of $p$, $p(\fA \sqcup_{rich} \fB)$ is given by 
\begin{eqnarray}
  & &\sum_{U\subseteq A\sqcup B: \Omega(U)} 
\left(\prod_{c \in A \sqcup B: \phi_\mathbf{X}(c,U)} \mathbf{X}\right)
\left(\prod_{c \in A\sqcup B:  \phi_\mathbf{Y}(c,U)} \mathbf{Y}\right) = \label{eq:two-prods} \\
  & &\sum_{U\subseteq A\sqcup B: \Omega(U)} 
         \left(\prod_{c \in A: \phi_\mathbf{X}(c,U)} \mathbf{X}\right) \left(\prod_{c \in B: \phi_\mathbf{X}(c,U)} \mathbf{X}\right)
	 \left(\prod_{c \in A: \phi_\mathbf{Y}(c,U)} \mathbf{Y}\right) \left(\prod_{c \in B: \phi_\mathbf{Y}(c,U)} \mathbf{Y}\right) 
	 \notag
\end{eqnarray}
Consider the summation in Equation (\ref{eq:two-prods}). It is a sum over all 
$U\subseteq A\sqcup B$ such that $\left\langle \fA \sqcup_{rich} \fB,U \right\rangle\models \Omega$.
There is a unique partition $(U\cap A) \sqcup (U\cap B) = U$.
We have that $\left\langle \fA \sqcup_{rich} \fB,U \right\rangle\models \Omega$
iff $\left\langle \fA, U\cap A \right\rangle \sqcup_{rich} \left\langle \fB,U\cap B \right\rangle\models \Omega$
by the definition of rich disjoint union. 

Let $\tau_{U}$ be the vocabulary of $\left\langle \fA, U\cap A \right\rangle$ and $\left\langle \fB, U\cap B \right\rangle$,
and let $\sigma_{U}$ be the vocabulary of 
$\left\langle \fA, U\cap A \right\rangle \sqcup_{rich} \left\langle \fB,U\cap B \right\rangle$. 
By Proposition \ref{prop:reform}, there exists $g_{\tau_{U}}$ such that 
\[
\theta_{\sigma_{U},q}(\left\langle \fA, U\cap A \right\rangle \sqcup_{rich} \left\langle \fB, U\cap B \right\rangle ) 
= g_{\tau_{U}}(\theta_{\tau_{U},q}(\left\langle \fA, U\cap A \right\rangle),\theta_{\tau_{U},q}(\left\langle \fB ,U\cap B \right\rangle))
\]
Hence we can write: \footnote{The notation $\models$ in $g_{\tau_{U}}(\theta_1,\theta_2)\models \Omega$ in the following formula denotes that
$g_{\tau_{U}}(\theta_1,\theta_2)$ entails $\Omega$. (Note that the notation $\models$ is usually used in this paper in expressions 
such as $\fA\models \phi$, where $\models$
denotes that the structure $\fA$ satisfies $\phi$.) }
\begin{gather} \label{eq:with-Theta}
\begin{array}{lll}
\displaystyle{
p(\fA \sqcup_{rich} \fB) 
   }& = & \displaystyle{\sum_{\substack{(\theta_1,\theta_2)\in \HinSet_{\tau_{U},q}\times\HinSet_{\tau_{U},q}:\\ g_{\tau_{U}}(\theta_1,\theta_2)\models \Omega}}
 \sum_{\substack{ U = U_1 \cup U_2
                  \\ U_1\subseteq A :\left\langle \fA, U_1\right\rangle \models \theta_1
                  \\ U_2\subseteq B :\left\langle \fB, U_2 \right\rangle \models \theta_2
                }
       }}
       \left(\prod_{c \in A: \phi_\mathbf{X}(c,U)} \mathbf{X}\right)
       \\
       & & \ \ \ \ \ \displaystyle{
\left(\prod_{c \in B: \phi_\mathbf{X}(c,U)} \mathbf{X}\right)
\left(\prod_{c \in A: \phi_\mathbf{Y}(c,U)} \mathbf{Y}\right)
\left(\prod_{c \in B: \phi_\mathbf{Y}(c,U)} \mathbf{Y}\right)
}
\end{array}
\end{gather}

Now consider e.g. the product
 $\prod_{c \in A: \phi_\mathbf{X}(c,U)} \mathbf{X}$. It is a product over all $c\in A$ such that 
\[
\left\langle \fA \sqcup_{rich} \fB,c,U \right\rangle\models \phi_\mathbf{X}\,.
\]
Again this condition can be written as 
a condition on a rich disjoint union of two structures:
$\left\langle \fA \sqcup_{rich} \fB,c,U \right\rangle\models \phi_\mathbf{X}$ iff 
$\left\langle \fA, U_1, c \right\rangle \sqcup_{rich} \left\langle \fB,U_2 \right\rangle\models \phi_\mathbf{X}$.

Let $\tau_{Uc}$ and $\sigma_{Uc}$ be the vocabularies of 
$\left\langle \fA, U_1, c \right\rangle$ and
$\left\langle \fA, U_1, c \right\rangle \sqcup_{rich} \left\langle \fB,U_2 \right\rangle$
respectively. 
The vocabulary of $\left\langle \fB,U_2 \right\rangle$ is $\tau_{U}$. 

Using Proposition \ref{prop:reform} there exists a function $g_{\tau_{Uc},\tau_{U}}$ such that 
\begin{eqnarray*}
 \theta_{\sigma_{Uc},q}(\left\langle \fA, U_1, c \right\rangle \sqcup_{rich} \left\langle \fB,U_2 \right\rangle) &=&
 g_{\tau_{Uc},\tau_{U}}(\theta_{\tau_{Uc},q}(\left\langle \fA, U_1, c \right\rangle),
\theta_{\tau_{U},q}(\left\langle \fB,U_2 \right\rangle))
\,.
\end{eqnarray*}
For every $\alpha\in\L_q(\tau_{U})$, let 
\[
\psi_{\mathbf{X},1,\alpha} = \bigvee_{\substack{\theta \in \HinSet_{\tau_{Uc},q}:\\
g_{\tau_{Uc},\tau_{U}}(\theta,\alpha)\models \phi_\mathbf{X} } } \theta
\]
$\psi_{\mathbf{X},1,\alpha}$ is an $\L_q$-formula. 

If $\left\langle \fB, U_2 \right\rangle \models \alpha$, then
\[
\prod_{\substack{c \in A:\\
\left\langle \fA \sqcup_{rich} \fB,c,U \right\rangle\models \phi_\mathbf{X}}} \mathbf{X} = 
\prod_{\substack{c \in A:\\
 \left\langle \fA, U_1, c \right\rangle \models \psi_{\mathbf{X},1,\alpha}}} \mathbf{X} 
\]
Similarly, there exist $\L_q$-formulas $\psi_{\mathbf{Y},1,\alpha}$, $\psi_{\mathbf{X},2,\alpha}$ and $\psi_{\mathbf{Y},2,\alpha}$ 
such that
if $\left\langle \fB, U_2 \right\rangle \models \alpha$, then
\begin{eqnarray*}
\prod_{\substack{c \in A:\\
\left\langle \fA \sqcup_{rich} \fB,c,U \right\rangle\models \phi_\mathbf{Y}}} \mathbf{Y} &=& 
\prod_{\substack{c \in A:\\
 \left\langle \fA, U_1, c \right\rangle \models \psi_{\mathbf{Y},1,\alpha}}} \mathbf{Y} 
\end{eqnarray*}
and if $\left\langle \fA, U_1 \right\rangle \models \alpha$, then
\begin{eqnarray*}
\prod_{\substack{c \in B:\\
\left\langle \fA \sqcup_{rich} \fB,c,U \right\rangle\models \phi_\mathbf{X}}} \mathbf{X} &=& 
\prod_{\substack{c \in B:\\
 \left\langle \fB, U_2, c \right\rangle \models \psi_{\mathbf{X},2,\alpha}}} \mathbf{X} 
\\
\prod_{\substack{c \in B:\\
\left\langle \fA \sqcup_{rich} \fB,c,U \right\rangle\models \phi_\mathbf{Y}}} \mathbf{Y} &=& 
\prod_{\substack{c \in B:\\
 \left\langle \fB, U_2, c \right\rangle \models \psi_{\mathbf{Y},2,\alpha}}} \mathbf{Y} \,.
\end{eqnarray*}

For every two $\L(\tau_{Uc})$-sentence $\varphi_{\mathbf{X}}$ and $\varphi_{\mathbf{Y}}$ and 
every Hintikka sentence $\varphi\in \HinSet_{\tau_{U},q}$,
let $\tp_{\varphi_{\mathbf{X}},\varphi_{\mathbf{Y}},\varphi}$ be the following $\L(\tau)$-polynomial:
\[
\tp_{\varphi_{\mathbf{X}},\varphi_{\mathbf{Y}},\varphi}(\fM)=
\sum_{U_\fM\subseteq M :\left\langle \fA, U_\fM \right\rangle \models \varphi }
\prod_{\substack{c \in M:\\
 \left\langle \fM, U_\fM, c \right\rangle \models \varphi_{\mathbf{X}}}} \mathbf{X}
\prod_{\substack{c \in M:\\
 \left\langle \fM, U_\fM, c \right\rangle \models \varphi_{\mathbf{Y}}}} \mathbf{Y}  
\]
Note that $\tp_{\varphi_{\mathbf{X}},\varphi_{\mathbf{Y}},\varphi}$ is a $\L$-polynomial over
the vocabulary $\tau$ with quantifier rank $q$. 
Then using Eq. (\ref{eq:with-Theta}) we have
\begin{eqnarray*}
  p(\fA \sqcup_{rich} \fB) 
   & = & \sum_{\substack{(\theta_1,\theta_2)\in \HinSet_{\tau_{U},q}\times\HinSet_{\tau_{U},q}:\\ 
		g_{\tau_{U}}(\theta_1,\theta_2)\models \Omega}}
         \tp_{\psi_{\mathbf{X},1,\theta_2},\psi_{\mathbf{Y},1,\theta_2},\theta_1}(\fA) 
	 \cdot \tp_{\psi_{\mathbf{X},2,\theta_1},\psi_{\mathbf{Y},2,\theta_1},\theta_2}(\fB)\,.
\end{eqnarray*}

Let $t=|\HinSet_{\tau_{U},q} \times \HinSet_{\tau_{U},q}|$ and
let $\pi:\{1,\ldots,t\}\to \HinSet_{\tau_{U},q} \times \HinSet_{\tau_{U},q}$
be a bijection.
For each $i=1,\ldots,t$, denote $\pi(i)=(\pi_1(i),\pi_2(i))$ and let
$p_i = \tp_{\psi_{\mathbf{X},1,\pi_2(i)},\psi_{\mathbf{Y},1,\pi_2(i)},\pi_1(i)}$ and 
$p_{i+t} = \tp_{\psi_{\mathbf{X},2,\pi_1(i)},\psi_{\mathbf{Y},2,\pi_1},\pi_2(i)}$. 
We set
\[
 Q(w_1,\ldots,w_t,v_1,\ldots,v_t) = 
\sum_{\substack{(\theta_1,\theta_2) \in \HinSet_{\tau_{U},q} \times \HinSet_{\tau_{U},q}:\\ g_{\tau_{U}}(\theta_1,\theta_2)\models \Omega}}
w_{\pi^{-1}(\theta_1,\theta_2)} \cdot v_{\pi^{-1}(\theta_2,\theta_1)}
\]
and the theorem follows.
The matrix $M=(m_{ij})$ is given explicitly as
\[
 m_{ij}=
 \begin{cases} 
  1, & \pi_1(i)=\pi_2(j)\mbox{, }\pi_2(i)=\pi_1(j)\mbox{, and }g_{\tau_{U}}(\pi(i))\models \Omega \\
  0, & \mbox{otherwise.}
 \end{cases}\eqno{\qEd}
\]

\begin{rem} If the statement of Theorem \ref{th:fv-disjoint-union} is changed such that
instead of rich disjoint union we take simple disjoint union, and correspondingly, 
$p$ has vocabulary $\tau$, then the theorem holds as well. Additionally,  there are polynomials $Q_1,\ldots,Q_{2t} \in \Z[w_1,\ldots,w_t,v_1,\ldots,v_t]$
such that 
\[
p_i(\fA \sqcup \fB) = Q_i\left(p_1(\fA),\ldots,p_t(\fA),p_{t+1}(\fB),\ldots,p_{2t}(\fB)\right)
\]
and $Q_i$ can be written in bilinear form. Here we use the fact that the number
of $\L(\tau)$-polynomials with a fixed quantifier rank bound on their formulas and any fixed set of indeterminates  is finite. This
follows from the fact that
the same is true for the number of formulas 
of any  fixed vocabulary and quantifier rank. 
\end{rem}

As a consequnce of Theorem \ref{th:fv-disjoint-union} we have:
\begin{thm}
Let $\cL$ be a tame fragment of $\SOL$ such that $\sqcup_{rich}$ is $\L$-smooth.
Let $p$ be a $\L(\sigma)$-polynomial in the form of Equation (\ref{eq:p}).  
Then the connection matrix $M(\sqcup_{rich},p)$ 
has finite rank over the field of rational functions with indeterminates $\mathbf{X}$ and $\mathbf{Y}$. 
\end{thm}

\begin{rem}
The entries of $M(\sqcup_{rich},p)$ are polynomial, and hence we consider the rank of $M(\sqcup_{rich},p)$ over the field of rational functions. 
Note that for every evaluation $p_{\mathbf{X}_0,\mathbf{Y}_0}(G)=p(G,\mathbf{X}_0,\mathbf{Y}_0)$ of $p$ such that $\mathbf{X}_0,\mathbf{Y}_0$ belong to some subfield $\mathbb{F}$ of $\mathbb{C}$,
the matrix $M(\Box,p_{\mathbf{X}_0,\mathbf{Y}_0})$ has finite rank over $\mathbb{F}$. 
\end{rem}

\subsection{Replacing \texorpdfstring{$\Z$}{Z} by arbitrary semirings} \label{subse:semi-rings}

Theorem \ref{th:fv-disjoint-union} remains true when interpret the $\L$-polynomials not over the ring $\Z$ with the standard $+$ and $\cdot$, but over other 
rings and semirings. 
 Stated explicitly, the following is true:
\begin{thm}[Bilinear Reduction Theorem for Semirings]
Let $\cL$ be a tame fragment of $\SOL$ such that $\sqcup_{rich}$ is $\L$-smooth.
Let $\mathcal{S}$ be a semiring. 
Let $p$ be a $\L(\sigma)$-polynomial in the form of Equation (\ref{eq:p}). 
There exist $t\in \mathbb{N}$, $\L(\tau)$-polynomials $p_1,\ldots,p_{2t}$ 
and a polynomial $Q\in \mathcal{S}[w_1,\ldots,w_t,v_1,\ldots,v_t]$ 
such that
\begin{renumerate}
 \item
  for any two $\tau$-structures $\fA$ and $\fB$
 \[
  p(\fA \sqcup_{rich} \fB) = Q\left(p_1(\fA),\ldots,p_t(\fA),p_{t+1}(\fB),\ldots,p_{2t}(\fB)
\right)\,.
 \]
 \item
  $Q(\bar{w}, \bar{v})$ can be written as a bilinear form
   $$\bar{w}^{tr} M \bar{v}$$
 where the matrix $M$ has as entries polynomials with coefficients from
in $\mathcal{S}$ and is independent of the $\tau$-structures. 
\end{renumerate}
\end{thm}

\noindent For instance, the above holds for $\mathcal{S}=\mathbb{R}\cup\{-\infty\}$ with $max$ and $+$ as the semiring's addition
and multiplication, respectively, cf. \cite{bk:Butkovic2010}. 
Note, however, that it is not immediately clear how to extend the finite rank theorem to this context, 
since ranks
of matrices over semirings can be defined in various ways.

The Bilinear Reduction Theorem for Semirings has found applications
in the theory of weighted (aka multiplicity) automata, cf. \cite{pr:LabaiMakowsky2013}.


\subsection{Sum-like operations} \label{subse:proof-sum-like}
Here we prove a Feferman-Vaught-type theorem for sum-like operations. 
Let $\rho$ and $\tau$ be vocabularies and let $\sigma$ be the vocabulary of the disjoint union of two $\tau$-structures.
\begin{thm}
Let $\cL$ be a tame fragment of $\SOL$ such that $\sqcup_{rich}$ is $\L$-smooth
and $\Box$ be an $\cL$-sum-like operation between $\tau$-structures.
Let $p$ be a $\L(\rho)$-polynomial in the form of Equation (\ref{eq:p}). 
There exists a polynomial $Q\in \Z[w_1,\ldots,w_t,v_1,\ldots,v_t]$ 
and $\L(\tau)$-polynomials $p_1,\ldots,p_{2t}$ such that
\begin{renumerate}
 \item 
  for any two $\tau$-structures $\fA$ and $\fB$
 \[
  p(\fA \Box \fB) = Q\left(p_1(\fA),\ldots,p_t(\fA),p_{t+1}(\fB),\ldots,p_{2t}(\fB)
\right)\,.
 \]
 \item
 $Q(\bar{w}, \bar{v})$ can be written as a bilinear form.
\end{renumerate}
\end{thm}
\begin{proof}
Let $\Phi$ be a scalar 
$\cL$-transduction from $\sigma_1$ structures to $\rho$ structures such that
\[
 \fA \Box \fB = \Phi^*(\fA \sqcup_{rich} \fB)\,. 
\]
Using Theorem \ref{th:fv-disjoint-union} it is enough to show that there exists a $\L(\sigma)$-polynomial
$r$ such that
\[
 p(\Phi^*(\fA \sqcup_{rich} \fB)) = r(\fA \sqcup_{rich} \fB)\,.
\]
Let $\phi_{univ}(v)$ be the formula of $\Phi$ which defines the universe of the transduction. 
Such $r$ is given by
\[
 r(\fM)=\sum_{\substack{U\subseteq M:
\\
 \forall v (v\in U \to \phi_{univ}(v)) \land \Phi^{\sharp}(\Omega)}} 
\left(\prod_{c \in M: \Phi^{\sharp}(\phi_\mathbf{X}(c,U))} \mathbf{X}\right) 
\left(\prod_{c \in M: \Phi^{\sharp}(\phi_\mathbf{Y}(c,U))} \mathbf{Y}\right)\,,
\]
where we consider $U$ and $c$ in $\Omega$, $\phi_\mathbf{X}$ and $\phi_\mathbf{Y}$ 
as free variables for the purpose of applying the map $\Phi^{\sharp}$. 
We use here that $\L$ is tame and therefore closed under transductions and containments. 
\end{proof}

\begin{thm}
Let $p$ be a $\L(\rho)$-polynomial in the form of Equation (\ref{eq:p}). 
Then the connection matrix $M(\Box,p)$ 
has finite rank over the field of rational functions with indeterminates $\mathbf{X}$ and $\mathbf{Y}$. 
\end{thm}
Again evaluations of $p$ have finite rank over the relevant subfield of $\mathbb{C}$.

\section{Proving Non-definability of \texorpdfstring{$\cL(\tau)$}{L(tau)}-invariants}
\label{se:solpolx}

Using the Finite Rank Theorems requires showing that connection matrices have infinite rank. 
In the case of numeric $\tau$-invariants and $\tau$-polynomials,
a simple technique for obtaining infinite rank is by showing that $f$ is $\Box$-maximizing or $\Box$-minimizing. 
We will also use other techniques. 

\subsubsection*{\texorpdfstring{$\Box$-maximizing and $\Box$-minimizing parameters}{Maximizing and minimizing parameters}}
We say a $\tau$-parameter $f$ is {\em $\Box$-maximizing} ({\em $\Box$-minimizing}) if
there exist an infinite sequence of non-isomorphic 
$\tau$-structures $\mathcal{A}_1,\mathcal{A}_2,\ldots,$ $\mathcal{A}_i,\ldots$
such that for any $i\not=j$, 
\[f(\mathcal{A}_i\Box\mathcal{A}_j) = \max\{f(\mathcal{A}_i), f(\mathcal{A}_j)\}\,.\]
Furthermore, if $f$ is unbounded on $\mathcal{A}_1,\mathcal{A}_2,\mathcal{A}_3,\ldots$ then 
$f$ is {\em unboundedly $\Box$-maximizing}. 
Analogously we define {\em (unboundedly) $\Box$-minimizing}. 

\begin{prop}
\label{p:maxmin}
If $f$ is a unboundedly $\Box$-maximizing ($\Box$-minimizing) $\tau$-parameter, then 
$M(\Box,f)$ has infinite rank. 
\end{prop}

\subsection{Non-definability: numeric \texorpdfstring{$\CMSOL(\tau)$}{CMSOL(tau)}-parameters}
Using Proposition \ref{p:maxmin}, Theorem \ref{maintheorem} and Proposition \ref{p:fol-1} we show that many $\tau$-parameters are not 
$\CMSOL$-definable:
\begin{prop}\label{p:max}
The following graph parameters are not 
$\CMSOL$-definable in the language of hypergraphs.
\\
\textsl{Spectral radius, chromatic number, 
acyclic chromatic number, arboricity, star chromatic number, clique number, Hadwiger number, Haj\'os number, 
tree-width, path-width, clique-width, edge chromatic number, total coloring number, Thue number, maximum degree, circumference,
longest path, maximal connected planar (bipartite) induced subgraph, boxicity,
minimal eigenvalue, spectral gap, girth, degeneracy,} and \textsl{minimum degree}.
\end{prop}
\begin{proof}
All these graph parameters $g$ are unboundedly $\sqcup$-maximizing or $\sqcup$-minimizing. 
\end{proof}

Variations of the notions of $\Box$-maximizing or $\Box$-minimizing $\tau$-parameters can 
also lead to non-definability results, e.g.:
\begin{prop}
The number of connected components (blocks, simple cycles, induced paths) of maximum (minimum) 
size is not $\CMSOL$-definable in the language of hypergraphs.
\end{prop}
\begin{proof}
Consider the connection matrix with respect to the operation of disjoint union 
of graphs $i \cdot K_i$ which consists of the disjoint union of $i$ cliques of size $i$.
We denote the number of connected components of maximum size in a graph $G$ by $\#_{\mathrm{max-cc}}(G)$. 
Then
\[
 \#_{\mathrm{max-cc}}(n K_n\sqcup m K_m)= 
\begin{cases}
 \max\{n,m\} & n\not=m \\
 n+m & n=m 
\end{cases}
\]
So $M(\sqcup,\#_{\mathrm{max-cc}})$ is of infinite rank. 
The other cases are proved similarly.
\end{proof}

In Section \ref{se:means} we discuss the rank of connection matrices of
various graph parameters based on averages, such as the average degree,
and many others. These are examples where the computation of the rank is a bit
more sophisticated.

\subsection{Non-definability: \texorpdfstring{$\CMSOL(\tau)$}{CMSOL(tau)}-polynomials}

Here we use the method of connection matrices for showing that (hyper)graph polynomials are not $\MSOL$-definable. 
Some of the material here is taken from the first author's thesis \cite{phd:Kotek}.
As examples we consider the polynomials $\chi_{rainbow}(G,k)$, $\chi_{mcc(t)}(G,k)$, and $\chi_{convex}(G,k)$,
which were defined in the introduction.

To show that none of $\chi_{rainbow}(G,k)$, $\chi_{mcc(t)}(G,k)$, or $\chi_{convex}(G,k)$ are $\CMSOL$-poly\-nomials
in the language of graphs,
and that neither $\chi_{rainbow}(G,k)$ nor $\chi_{convex}(G,k)$ are $\CMSOL$-polynomials 
in the language of hypergraphs,
we prove the following general proposition: 

\begin{lem}
\label{lm:non-def-gp}
Given a  $\tau$-parameter $p$, a binary operation $\Box$ on $\tau$-structures and
an infinite sequence of non-isomorphic $\tau$-structures
$\mathcal{A}_i, i \in \N$,
let $f:\N\to \N$ be an unbounded function such that one of the following occurs:
\begin{renumerate}
 \item \label{item:1}
for every $\lambda\in \N$, 
$p(\mathcal{A}_i\Box \mathcal{A}_j,\lambda)=0$ iff $i+j>f(\lambda)$.
 \item \label{item:2}
for every $\lambda\in \N$, 
$p(\mathcal{A}_{2i}\Box \mathcal{A}_{2j+1},\lambda)=0$ iff $i+j>f(\lambda)$.
\end{renumerate}
Then the connection matrix $M(\Box,p)$ has infinite rank.
\end{lem}
\begin{proof}
Let $\lambda\in \N$ and let $p_\lambda$ be the graph parameter given by $p_\lambda(G)=p(G,\lambda)$. 
If (\ref{item:1}) holds, consider the restriction $N$ of the connection matrix $M(\Box,p_\lambda)$ 
to the rows and columns corresponding to $\mathcal{A}_i$, 
$0\leq i\leq f(\lambda)-1$. 
If (\ref{item:2}) holds, consider the restriction $N$ of $M(\Box,p_\lambda)$ to to the rows corresponding to
$\mathcal{A}_{2i}$ and the columns corresponding to $\mathcal{A}_{2i+1}$, $0\leq i\leq f(\lambda)-1$. 
In both cases $N$ is a  finite triangular matrix with non-zero diagonal. Hence the rank of 
$M(\Box,p_\lambda)$ is at least $f(\lambda)-1$.

Using that $f$ is unbounded, we get that $M(\Box,p)$ contains infinitely 
many finite sub-matrices with ranks which tend to infinity.
Hence, the rank of $M(\Box,p)$ is infinite, 
\end{proof}

We now use Lemma \ref{lm:non-def-gp} to compute connection matrices 
where $\Box$ is the disjoint union $\sqcup$, 
the $1$-sum $ \sqcup_1$ or the join $\bowtie$.
In addition to the graph polynomials introduced in the introduction, we consider the definability
of some {\em $\cP$-polynomials} denoted $\chi_{\cP}(G,k)$ where $\cP$ is a graph property. 
$\chi_{\cP}(G,k)$ counts vertex $k$-colorings $f:V(G)\to[k]$ such that for every color $c\in [k]$, 
the graph $f^{-1}(c)$ induced by the vertices colored $c$ under $f$ belongs to $\cP$. 
Such colorings were introduced by F. Harary, cf. \cite{pr:Harary85,ar:HararyHsu1991,ar:BrownCorneil1987,ar:BrownCorneil1991}. 
Several of the colorings presented so far are $\cP$-
colorings for appropriate choices of $\cP$.

\begin{prop}
The following connection matrices have infinite rank:
\begin{renumerate}
\item
$M(\bowtie, \chi(G,k))$;
\item
For every $t >0$ the matrix
$M(\bowtie, \chi_{mcc(t)}(G,k))$;
\item
$M(\bowtie, \chi_{v-acyclic}(G,k));$
\item
$M(\bowtie, \chi_{\cP_{Bipartite}}(G,k))$;
\item
$M(\bowtie, \chi_{\cP_{Forest}}(G,k))$;
\item
$M(\bowtie, \chi_{\cP_{Tree}}(G,k))$;
\item
$M(\bowtie, \chi_{\cP_{Planar}}(G,k))$;
\item
$M(\bowtie, \chi_{\cP_{3-regular}}(G,k))$;
\item
$M(\sqcup_1, \chi_{rainbow}(G,k))$;
\item
$M(\sqcup, \chi_{convex}(G,k))$;
\item
For every $t >0$ the matrix
$M(\sqcup_1, \chi_{t-improper}(G,k))$;
\item
$M(\sqcup_1, \chi_{non-rep}(G,k))$;
\item
$M(\sqcup, \chi_{harm}(G,k))$.
\end{renumerate}
\end{prop}
\begin{proof}~
\begin{renumerate}
\item
For $\chi(G,k)$ we use the fact that join of cliques is again a clique, $K_{i} \bowtie K_j = K_{i+j}$ and 
that $\chi(K_r,k)=0$ iff $r> k$. 
\item
In a similar fashion, for $\chi_{mcc(t)}(G,k)$ we use cliques and 
that $\chi_{mcc(t)}(K_r,k)=0$ iff $r> k t$. 
\item 
For $\chi_{v-acyclic}(G,k)$ we use cliques. We have $\chi_{v-acyclic}(K_n,k)=0$ iff $n>k$. 
\item
For $\chi_{\cP_{Bipartite}}(G,k)$ we use cliques and the fact that $\chi_{\cP_{Bipartite}}(K_i,k)=0$ iff $i>2k$. 
\item
For $\chi_{\cP_{Forest}}(G,k)$ we use cliques and the fact that $\chi_{\cP_{Forest}}(K_i,k)=0$ iff $i>2k$. 
\item
For $\chi_{\cP_{Tree}}(G,k)$ we use cliques and the fact that $\chi_{\cP_{Tree}}(K_i,k)=0$ iff $i>2k$. 
\item
For $\chi_{\cP_{Planar}}(G,k)$ we use cliques and the fact that $\chi_{\cP_{Planar}}(K_i,k)=0$ iff $i>4k$. 
\item
For $\chi_{\cP_{3-regular}}(G,k)$ we use cliques and the fact that $\chi_{\cP_{3-regular}}(K_i,k)=0$ iff $i>3k$. 
\item
For $\chi_{rainbow}(G,k)$, we use that the $1$-sum of paths with one end labeled is again a path 
with
$P_{i} \sqcup_1 P_j = P_{i+j-1}$ and 
that $\chi_{rainbow}(P_r,k)=0$ iff $r> k+3$. 
\item
For $\chi_{convex}(G,k)$, we use edgeless graphs and disjoint union $E_{i} \sqcup E_j = E_{i+j}$ and 
that $\chi_{convex}(E_r,k)=0$ iff $r> k$. 
\item
For $\chi_{t-improper}(G,k)$ we use cliques and that $\chi_{t-improper}(K_i \sqcup_1 K_j,k)=0$ 
iff $\left\lceil\frac{i+j-2}{k} \right\rceil>t$
\item
For $\chi_{non-rep}(G,k)$ we use that $1$-sum of stars $S_n$ is again a star. We have \linebreak
$\chi_{non-rep}(S_n,k)=0$ iff $n>k$. 
\item
For $\chi_{harm}(G,k)$ we use the graphs $nK_2$ consisting of $n$ disjoint edges. We have
$\chi_{harm}(nK_2,k)=0$ iff $n>\binom{k}{2}$. 
\end{renumerate}
\end{proof}

\begin{cor}~
\begin{renumerate}
\item
$\chi(G,k)$, $\chi_{mcc(t)}(G,k)$ (for any fixed $t>0$), $\chi_{v-acyclic}(G,k)$, $\chi_{\cP_{Bipartite}}(G,k)$, \linebreak
$\chi_{\cP_{Forest}}(G,k)$, $\chi_{\cP_{Tree}}(G,k)$, $\chi_{\cP_{Planar}}(G,k)$ and
$\chi_{\cP_{3-regular}}(G,k)$
are not $\CMSOL$-definable in the language of graphs.
\item 
$\chi_{rainbow}(G,k)$, $\chi_{convex}(G,k)$, $\chi_{non-rep}(G,k)$, and $\chi_{t-improper}(G,k)$ (for any fixed $t>0$)
 are not $\CMSOL$-definable in the languages of
graphs and hypergraphs.
\end{renumerate}
\end{cor}
\begin{proof}
(i) The join is 
$\CMSOL$-sum-like and $\CMSOL$-smooth for graphs (but not for hypergraphs). 
(ii) The $1$-sum and the disjoint union are 
$\CMSOL$-sum-like and $\CMSOL$-smooth for graphs and hypergraphs.
\end{proof}

\subsection{Non-definability: Means and \texorpdfstring{$\CMSOL(\tau)$}{CMSOL(tau)}}
\label{se:means}
Maximization and minimization can sometimes be thought of as a type of mean.
For example, the maximal (minimal) degree of a graph is obtained from the generalized mean 
\[
 \left(\frac{\sum_{v\in V(G)} \mathrm{degree}(v)^p}{|V(G)|}\right)^{\frac{1}{p}}
\]
when $p$ tends to $\infty$ ($-\infty$).

Other instances of the generalized mean also lead to infinite connection matrices ranks.
In particular, below we show examples for $p=1$, $p=2$ and $p=-1$. 
The following Lemma will be useful:
\begin{lem}~\label{lem:rankprop}
Let $\Box$ be any binary operation between $\tau$-structures. 
\begin{renumerate}
\item The $\Box$-connection matrix of a linear combination of $\tau$-invariants with $\Box$-connection matrices of
finite rank is of finite rank.
 \item The $\Box$-connection matrix of a finite product of $\tau$-invariants with $\Box$-connection matrices of finite rank
is of finite rank.
 \end{renumerate}
\end{lem}
\begin{proof}~
\begin{renumerate}
\item Follows from the sub-additivity of the rank of matrices.
 \item It is enough to prove the claim for the product of two graph invariants,
$f$ and $g$. Denote the $\Box$-connection matrices of $f$ and $g$
by $M$ and $N$ respectively. Since $M$ and $N$ are of finite rank,
there exists $t\in \N$ such that for every $i\in \N$, 
the row $M_i$ is a linear combination of the rows $M_1,\ldots,M_t$
and $N_i$ is a linear combination of the rows $N_1,\ldots,N_t$. 
Hence, for every $i\in \N$ there exist $c_1,\ldots,c_t,d_1,\ldots,d_t$ such that
for every $j\in \N$,
\[
 M_{i,j}=\sum_{r\leq t} c_r M_{r,j} \mbox{\ \ \ and \ \ \ } 
 N_{i,j}=\sum_{s\leq t} d_s N_{s,j}\,. 
\]
Hence, 
\[
 M_{i,j} N_{i,j}=\sum_{r,s\leq t} c_{r} d_{s} M_{r,j} N_{s,j}\,.
\]
So, $M(\Box,f \cdot g)$ is spanned by $t^2$ rows.
\end{renumerate}
\end{proof}

Let $\mathrm{avg\,val}(G)$ denote the average degree: 
\[
 \mathrm{avg\,val}(G) = \frac{\sum_{v\in V(G)} \mathrm{degree}(v) }{|V(G)|}\,.
\] 
\begin{prop}[Arithmetic mean]
The rank of $M(\sqcup,\mathrm{avg\,val})$ is infinite. 
\end{prop}
\begin{proof}
To see this,
look at the connection matrix $M(\sqcup,\mathrm{avg\,val})$, whose entries can be expressed as follows:
$$M(\sqcup,\mathrm{avg\,val})_{i,j}= 2 \frac{|E_i| + |E_j| }{|V_i| + |V_j|}.$$
The sub-matrix of $M(\sqcup,\mathrm{avg\,val})$ which consists only of rows and columns
corresponding to graphs with exactly one edge is 
the Cauchy matrix $(\frac{2}{i+j})_{i,j}$, hence the rank of $M(\sqcup,\mathrm{avg\,val})$ is infinite.
\end{proof}
We can prove a similar statement for other graph parameters given as arithmetic means:
\begin{prop}\label{p:averages}
The following arithmetic means have $\sqcup$-connection matrices of infinite rank:
\begin{renumerate}
 \item For every $i\in \N$, the average size of the $i$-th neighborhood of vertices $v\in V(G)$,  \linebreak
$|\{ u \mid 0\leq \mathrm{distance}(u,v)\leq i\}|$.
 \item Average number of simple cycles in which vertices $v\in V(G)$ occur.
 \item Average number of triangles in which vertices $v\in V(G)$ occur.
 \item Average size of connected component in which vertices $v\in V(G)$ occur.
\item Average number of edges incident to an edge $e\in E(G)$.
\item Average number of cycles which include $e\in E(G)$.
\item For every $i\geq 2$, the average number of $i$-paths which include $e\in E(G)$. 
\end{renumerate}
\end{prop}

\noindent Let the quadratic mean of the degrees of vertices $v\in V(G)$ be
\[
 \mathrm{q\,avg\,val}(G)=\left(\frac{\sum_{v\in V(G)} \mathrm{degree}(v)^2 }{|V(G)|}\right)^{\frac{1}{2}}\,.
\]
\begin{prop}[Quadratic mean]
The rank of $M(\sqcup,\mathrm{q\,avg\,val})$ is infinite.
\end{prop}
\begin{proof}
$(\mathrm{q\,avg\,val}(G))^2$ has infinite connection matrix rank with respect to $\sqcup$
by looking again at graphs with exactly one edge. Again, $M(\sqcup,\mathrm{q\,avg\,val}^2)$ 
has entries $M_{i,j}=\frac{2}{i+j}$. Hence, by Lemma \ref{lem:rankprop},
$M(\sqcup, \mathrm{q\,avg\,val}(G))$ has infinite rank.
\end{proof}

With similar proofs, we have:
\begin{prop}
The corresponding quadratic means to those in Proposition \ref{p:averages} have $\sqcup$-connection matrices of infinite rank.  
\end{prop}

Let the harmonic mean of the degrees
of vertices $v\in V(G)$ be
\[
 \mathrm{h\,avg\,val}(G)=\frac{|V(G)|}{\sum_{v\in V(G)} \frac{1}{\mathrm{degree}(v)}}\,.
\]
We only consider graphs with no isolated vertices in order to avoid devision by zero. 
\begin{prop}[Harmonious mean]
 The rank of $M(\bowtie,\mathrm{h\,avg\,val})$ is infinite. 
\end{prop}
\begin{proof}
Note that $K_{n,m}=E_n \bowtie E_m$ and 
consider 
\begin{eqnarray*}
\mathrm{h\,avg\,val}(K_{n,m}) 
&=& \frac{n+m}{\frac{n}{m} + \frac{m}{n}} \\
&=& \frac{1}{n^2 + m^2}\cdot nm (n+m) 
\end{eqnarray*}
For every function $f:\N\times\N\to \Q$, let
$N(f)$ be the matrix such that the entry in row $i$ and column $j$ is $f(i,j)$. 
Assume $M(\mathrm{h\,avg\,val},\bowtie)$ is of finite rank. 
Then $N_1 = N\left(\frac{ij(i+j)}{i^2+j^2}  \right)$ is of finite rank.
Clearly, $N_2 = N\left(\frac{j(i+j)}{i^2+j^2} \right)$ is also of finite rank,
since $N_2$ is obtained from $N_1$ by multiplying each row $i$ of $N_1$ by a scalar $\frac{1}{i}$.  
The matrix $N\left(\frac{1}{j}\right)$ is of row rank $1$ because all of its rows are equal. So,
$N_3 = N\left(\frac{i+j}{i^2+j^2} \right)$ is of finite rank by Lemma \ref{lem:rankprop}.
The matrix $N(i+j)$ is of row rank $2$, since it is spanned by the vectors
$(1\ 2\ 3\ 4\ \cdots)$ and $\mathbf{1}$.  
So, 
$N_4 = N\left(\frac{(i+j)^2}{i^2+j^2} \right)$ is of finite rank again by Lemma \ref{lem:rankprop}.
Now notice
\[           
\frac{(i+j)^2}{i^2+j^2}   = 
1 + \frac{2ij}{i^2+j^2}  \,.
\]
Hence, we have that $N_5 =N\left(\frac{1}{i^2+j^2}\right)$ 
is of finite rank, but $N_5$ is a Cauchy matrix and is therefore of infinite rank, in contradiction. 
\end{proof}

\subsection{Non definability: \texorpdfstring{$\CFOL(\tau)$}{CFOL(tau)}-parameters}
Here we discuss definability and non-definability of $\tau$-parameters in $\CFOL$. Recall from Section \ref{se:solpol} that
formulas in $\CFOL$-parameters may not have second order variables. Examples of $\CFOL$-parameters $p(G)$, $G=(V,E)$, in the vocabulary of graphs:
\[\begin{array}{llllll}
 |V| & =& \displaystyle  \sum_{v \in V} 1 \ \ \ \ \ \ \ \ \ \ \  \ \ \ \ \ \ \ \ \ \  \ \ \ \ \ \ \ \ \ \ \ \ \  \ &  |\mathit{Apex(G)}| &=&\displaystyle  \sum_{v \in V : \forall x ((x \approx v) \lor E(x,v))} 1
\\
 |E| & = & \sum_{v,u\in V: E(u,v) \land (u<v)} 1 & \mathit{odd-deg}(G) &= & \sum_{x\in V: \dd_{2,1}x E(y,x)} 1
\end{array}
\]
where $\mathit{Apex}(G)$ is the set of vertices of $G$ adjacent to all other vertices, 
$\mathit{dgen}(G)$ is the generating function 
 of the degrees of vertices in $G$, and $\mathit{odd-deg}(G)$
is the number of vertices of odd degree. 

On the other hand, we can use Theorem \ref{maintheorem-fo} to show non-definability of $\CFOL$-para\-meters:
\begin{prop}\label{prop:fo-non-def-poly} The following are not $\CFOL$-definable:
\begin{renumerate}
 \item The number $spn-f(G)$ of  spanning forests
 \item The number $spn-t(G)$ of  spanning trees. 
 \item The number $cyc(G)$ of  cycles in $G$. 
 \item The number $k(G)$ of connected components in $G$. 
 \item The tree-width $tw(G)$ of $G$. 
 \item The number $blk(G)$ of blocks of $G$. 
\end{renumerate}
\end{prop}

\begin{proof}
We use the constructions from Section \ref{se:fol}. 
\begin{renumerate}
 \item The number of spanning forests $spn-f$ in the graph from Figure \ref{fg:Forests}, the graph obtained by applying $\Phi_F$ to two directed paths of length $n_1$ and $n_2$,
is $1$ if $n_1\not=n_2$, and is $n_1$ if $n_1=n_2$. Hence, $M(\Phi_F,spn-f)$ has infinite rank. Since $\Phi_F$ is a $\CFOL$-product-like operation, 
$spn-f$ is not $\CFOL$-definable. 
\end{renumerate}
The other cases are similar. We describe them shortly. 
 \begin{renumerate}
\setcounter{enumi}{1}
\item We use $\Phi_T$ from Figure \ref{fg:Trees}. The graph is connected iff $n_1\leq n_2$, implying that the connection matrix of the number of spanning trees is zero below the diagonal and non-zero otherwise, so has infinite rank.
\item We use $\Phi_F$ from  Figure \ref{fg:Forests} again. We have $cyc(G) =1$ iff $n_1=n_2$, otherwise $cyc(G)=0$. Hence the connection matrix of $cyc(G)$  
has $2$ on the diagonal and $1$ otherwise. Consequently it has infinite rank. 
\item We use $\Phi_T$ from Figure \ref{fg:Trees}. We have $k(G) = 1$ if $n_1\leq n_2$ and $k(G)= n_1-n_2+1$ otherwise. Hence, the connection matrix of $k(G)$
so has infinite rank. 
\item We use $\Phi_P$ from Figure \ref{fg:Planar}. We have $tw(G)=4$ iff $n_1\not=n_2$, and otherwise $tw(G)=5$. 
\item We use $\Phi_B$ from Figure \ref{fg:Bridgeless}. We have $blk(G)=1$ if $n_1<n_2+1$, $blk(G)=2$ if $n_1=n_2+1$ and otherwise $blk(G)=3$. 
\end{renumerate}
\end{proof}

\section*{Acknowledgements}
 We are grateful for the detailed feedback from the referees which allowed us to improve the presentation of the paper.




\end{document}